\documentclass[journal]{IEEEtran}

\usepackage{amsmath,amsfonts}
\usepackage{algorithm}
\usepackage{algpseudocode}
\usepackage{array}
\usepackage{textcomp}
\usepackage{stfloats}
\usepackage{url}
\usepackage{verbatim}
\usepackage{graphicx}
\usepackage{cite}
\usepackage{subcaption}
\usepackage{amssymb}   
\usepackage[colorlinks=true, linkcolor=blue, citecolor=blue, urlcolor=blue]{hyperref}

\usepackage{amsthm}
\newtheorem{thm}{Theorem}
\newtheorem{assumption}{Assumption}
\newtheorem{corollary}{Corollary}

\newtheorem{proposition}{Proposition}

\usepackage{xcolor}
\definecolor{myLightRed}{rgb}{1,0.3,0.3}
\definecolor{myLightBlue}{rgb}{0.2,0.2,1}
\definecolor{myLightPurple}{rgb}{1,0.2,1}

\definecolor{myLightRed}{rgb}{0,0,0}
\definecolor{myLightBlue}{rgb}{0,0,0}
\definecolor{myLightPurple}{rgb}{0,0,0}

\begin{document}

\title{\textcolor{myLightRed}{Ambiguity Function Analysis of OFDM Signals \\ With Pilots and Data Payloads}}


\author{Jialin Wu, Fan Liu,~\IEEEmembership{Senior Member,~IEEE,} Ying Zhang,~\IEEEmembership{Graduate Student Member,~IEEE,} \\Yifeng Xiong,~\IEEEmembership{Member,~IEEE,} Jie Yang,~\IEEEmembership{Member,~IEEE,} Kawon Han,~\IEEEmembership{Member,~IEEE,} and Shi Jin,~\IEEEmembership{Fellow,~IEEE}
\thanks{This work was supported in part by the Mobile Information Networks-National Science and Technology Major Project under Grant 2025ZD1302000, and in part by the National Natural Science Foundation of China (NSFC) under Grant 62522107 and Grant 62331023. An earlier version of this paper \cite{wu2026ambiguity} has been submitted to the 2026 IEEE Globecom Workshops. \it(Corresponding author: Fan Liu.)\rm}
\thanks{Jialin Wu, Fan Liu and Shi Jin are with the National Mobile Communications Research Laboratory, Southeast University, Nanjing 210096, China (e-mail: jialinw@seu.edu.cn; fan.liu@seu.edu.cn; jinshi@seu.edu.cn).}
\thanks{Ying Zhang is with the School of Automation and Intelligent Manufacturing, Southern University of Science and Technology, Shenzhen 518055, China (email: zhangying2024@mail.sustech.edu.cn).}
\thanks{Yifeng Xiong is with the School of Information and Electronic Engineering, Beijing University of Posts and Telecommunications, Beijing 100876, China (e-mail: yifengxiong@bupt.edu.cn).}
\thanks{Jie Yang is with the Frontiers Science Center for Mobile Information Communication and Security and the Key Laboratory of Measurement and Control of Complex Systems of Engineering, Ministry of Education, Southeast University, Nanjing 210096, China (e-mail: yangjie@seu.edu.cn).}
\thanks{Kawon Han is with the Department of Electrical Engineering, Ulsan National Institute of Science and Technology (UNIST), Ulsan 44919, South Korea (e-mail: kawon.han@unist.ac.kr).}
}

\maketitle

\begin{abstract}
\textcolor{myLightRed}{
Practical orthogonal frequency division multiplexing (OFDM) communication frames contain both deterministic pilots and random data payloads, motivating the joint ambiguity function (AF) analysis of the two components when the entire frame is reused for integrated sensing and communication (ISAC). This paper characterizes two discrete AF formulations for different Doppler regimes, namely the discrete periodic AF (DP-AF) and fast-slow-time AF (FST-AF), and derives closed-form expressions for their expected squared values. For the FST-AF, the expected sidelobe level (ESL) is uniform over the delay-Doppler plane and depends only on the pilot count, constellation kurtosis and total number of time-frequency resources, but not on the pilot symbols or pattern. For the DP-AF, we establish attainable lower and upper ESL bounds and show that no pilot design can minimize all sidelobes simultaneously. We further prove that attaining the lower bound at non-zero Doppler requires a periodic pilot pattern, while equally spaced chirp pilots, including Zadoff-Chu (ZC) sequences, maximize the numbers of sidelobes attaining the lower and upper bounds simultaneously. Two representative ZC pilot patterns widely encountered in communication frames are then examined: contiguous placement produces delay-Doppler ridges described by squared Dirichlet kernels, whereas equally spaced placement generates periodic peak-and-notch structures. Both regular patterns exhibit pronounced high sidelobes, suggesting that communication-oriented pilot patterns should be re-designed for delay-Doppler estimation in the context of ISAC. Numerical results validate the analysis and show that irregular pilot placement can suppress high sidelobes and improve target estimation performance.
}
\end{abstract}


\newpage

\section{Introduction}
\IEEEPARstart{T}{he} development of 6G networks targets emerging applications like \textcolor{myLightPurple}{low-altitude economy} and autonomous driving \cite{rodr20256g}, demanding both seamless connectivity and accurate environmental perception. Integrated sensing and communication (ISAC) meets this need by merging sensing with communication, reusing hardware and spectrum for higher efficiency and mutual enhancement \cite{niu2024interference}. A central objective of ISAC is to develop dual-purpose waveforms that support both communication and sensing functions \cite{bazzi2025mutual}. In general, current waveform design methodologies fall into the following classes: communication-centric, radar-centric, and joint approaches \cite{du2025full}. Communication-centric methods repurpose 
conventional 
physical-layer waveforms, including orthogonal frequency division multiplexing (OFDM), to perform sensing tasks. Radar-centric approaches integrate communication information into radar waveforms. Joint designs aim to create new waveforms that jointly serve both functions. Among these, the communication-centric paradigm offers the advantage of infrastructure compatibility, allowing deployment with minimal modification \cite{galappaththige2026cell}, while preserving communication performance.

Current communication-centric ISAC designs often use reference signals \textcolor{myLightBlue}{including the channel state information reference signal (CSI-RS) \cite{yadav2026bayesian} and the positioning reference signal (PRS) \cite{RSVTC}} to support sensing. 
The structured pilot signals \textcolor{myLightRed}{are typically composed by well-designed sequences}, which give favorable correlation behaviors for basic estimation \cite{zhu2023pilot}. However, the time-frequency resources they occupy are  limited to about 10\%, leaving the bulk of the remaining resources unexploited for sensing purposes. Due to the increasing demand for higher sensing accuracy, it is \textcolor{myLightRed}{necessary} to also utilize the remaining resources that carry random \textcolor{myLightPurple}{data payloads},
so as to boost sensing performance and spectrum utilization \cite{11572999}.

Towards that end, recent efforts have started to explore the design of random communication waveforms, with OFDM gaining attention due to its widespread deployment in existing standards \cite{han2026constellation} and effective ranging sidelobe suppression. A representative study derived the average squared auto-correlation function (ACF) under random signaling to assess the statistical capability of range estimation \cite{CP-OFDM}. The results indicated that for all orthogonal waveforms that include a cyclic prefix (CP), OFDM yields the lowest sidelobe levels for the constellation families of quadrature amplitude modulation (QAM) and phase shift keying (PSK). Extended work confirmed that OFDM continues to exhibit the lowest ranging sidelobes under Nyquist pulse shaping \cite{iceberg}, underscoring the potential of utilizing OFDM for sensing under random ISAC signaling.

While the ACF offers a fundamental measure for ranging performance, it only captures the behavior of the ambiguity function (AF) at the zero-Doppler cut and therefore fails to reveal the full delay-Doppler behavior.
The AF serves as a foundational tool for analyzing the matched-filter response of a waveform across the delay-Doppler domain \cite{rihaczek1967delay}, making it essential for evaluating resolution and interference in dynamic sensing scenarios. The AF of communication waveforms has attracted considerable research interest\textcolor{myLightBlue}{.}
In \cite{dayarathna2024otfs}, the AF of the orthogonal time-frequency space (OTFS) waveform was characterized by approximating its mean and variance. The average squared 
AFs of affine frequency division multiplexing (AFDM) and OFDM waveforms were derived in \cite{ni2026ambiguity}. Additionally, a recent study \cite{zhang2025discrete} sought to establish a unified statistical characterization for random waveforms, including OFDM, using two formulations of discrete AFs: the discrete periodic AF (DP-AF) and the fast-slow-time AF (FST-AF). The DP-AF captures the periodic self-convolution over the delay-Doppler grid for CP-equipped \textcolor{myLightRed}{signals under block} transmissions. The FST-AF, in contrast, assumes that the Doppler shift is constant for the duration of each fast-time block and varies only across slow-time pulses, which decouples the two dimensions of delay and Doppler. The FST-AF approximates the DP-AF under the small-Doppler assumption, \textcolor{myLightBlue}{yielding a compact expression for AF characteristics and 
providing an insight into the delay-Doppler sidelobe structure that is free of inter-carrier interference (ICI).}
The findings in \cite{zhang2025discrete} indicate that, in the DP-AF formulation, there does not exist a waveform capable of reaching the lowest expected sidelobe level across an arbitrary two-dimensional area of the delay-Doppler plane. Conversely, within the FST-AF formulation, OTFS achieves the lowest sidelobes for super-Gaussian constellations, while OFDM becomes optimal for sub-Gaussian constellations including commonly used QAM and PSK, highlighting the benefits of OFDM for sensing applications in communication-centric ISAC.




\textcolor{myLightRed}{Practical OFDM communication frames typically comprise both random data symbols and deterministic pilots, making it necessary to characterize their joint ambiguity behavior when the entire communication waveform is exploited for sensing. 
Importantly, the standardization of ISAC is actively underway in the 3rd generation partnership project (3GPP), where a key open question is whether sensing should leverage a dedicated pilot or reuse the pilots already embedded in the communication frame \cite{3gpp_pilot}. This debate directly hinges on how well the embedded pilots can support sensing together with the coexisting data payload, which makes the characterization of pilot-embedded OFDM waveforms a timely and practically important problem.
Although previous studies have characterized the statistical ambiguity behavior of OFDM signals
composed solely of random data, the analysis has not yet been extended to the practical frame structures, which have already been exploited for sensing in recent studies.} An affine-precoded superimposed pilot scheme that jointly harnesses pilots and data for OFDM sensing was introduced in \cite{gupta2024affine}. Non-uniform pilot patterns were designed to suppress the pilot grating lobes for OFDM sensing \cite{bouziane2026optimized}. An AF analysis of pilot-assisted AFDM frames was presented in \cite{zhang2025afdm}. A multi-antenna approach in \cite{Both_Pilots_and_Data_Payloads} characterized the mean square error when utilizing pilot sequences and data symbols for sensing. Nevertheless, the statistical ambiguity behavior of sensing strategies that process pilot and data components in OFDM signals has not been systematically characterized. 

\textcolor{myLightBlue}{Whereas existing AF analyses assume OFDM signals composed of independent and identically distributed data, the pilot-embedded signals considered here mix random data with deterministic pilots that are placed on distinct time-frequency resources. This heterogeneity renders the transmitted symbols no longer identically distributed and couples the pilot structure and the data statistics within the AF. }\textcolor{myLightRed}{This gives rise to cross-correlation terms that cannot be handled by existing frameworks and therefore necessitate a novel analytical decomposition.}
A rigorous characterization is needed to reveal how the design of pilots and that of communication symbols respectively shape the sidelobe behavior and affect the overall sensing performance, offering insights \textcolor{myLightRed}{to} the joint design of pilot and data resources in OFDM ISAC systems.






To fill the research gap, this work 
\textcolor{myLightRed}{develops an AF analysis framework} for OFDM signals that contain both random data and deterministic pilot symbols, covering the DP-AF and FST-AF.
We derive closed-form expected squared expressions, reveal their dependence on key system parameters, establish achievable bounds of expected sidelobe levels, and examine typical pilot patterns when using Zadoff-Chu (ZC) sequences as pilots. The major contributions are outlined below.

\begin{itemize}
\item We present analytical expressions for the expected squared DP-AF and FST-AF of OFDM signals containing data payloads and unit-amplitude pilots. For each fixed signal length, the zero-Doppler sidelobes of the expected squared DP-AF are governed by the constellation kurtosis and pilot count, while the remaining sidelobes depend on the pilot pattern and symbols. For FST-AF, given the number of time-frequency resources, the expected sidelobe levels are all identical and are determined by the pilot count and constellation kurtosis.

\item For the expected squared DP-AF, we demonstrate that for any non-zero Doppler shift, the sum of the sidelobe levels over delay \textcolor{myLightPurple}{is 
independent} of pilot design, and that cyclically shifting the pilot pattern leaves the expected squared DP-AF unchanged. More importantly, attainable upper and lower bounds of expected sidelobe levels are established, \textcolor{myLightPurple}{and we show that no pilot design can minimize all sidelobes in a two-dimensional delay-Doppler region simultaneously. We further prove that attaining the lower or upper bound at a non-zero Doppler requires a periodic pilot pattern}, and the maximum number of bound-achieving sidelobes is revealed to be realized by equally spaced chirp sequences, including ZC sequences.




\item We then study the expected squared DP-AF when ZC sequences are used as pilots, with two typical pilot patterns. Equally spaced placement produces a comb-like ambiguity profile where sidelobes alternate between the upper and lower \textcolor{myLightPurple}{bounds.} 
Contiguous placement yields a squared Dirichlet-kernel structure with delay-Doppler \textcolor{myLightPurple}{ridges.} 
\textcolor{myLightPurple}{Both communication-oriented regular patterns induce pronounced high sidelobes. 
To re-design the pilot pattern in the context of ISAC, 
we adopt a randomly placed pattern to suppress the high sidelobes and improve target estimation
performance.}
\end{itemize}

The remainder of this paper proceeds as follows. Section \ref{section_2} formulates the pilot-embedded OFDM signals and defines the DP-AF and FST-AF. In Section \ref{section_3}, we derive expressions for the expected squared AFs, and analyze how pilots influence the ambiguity behavior. Section \ref{section_4} presents a case study where ZC sequences serve as pilots, comparing different pilot patterns. We provide simulation results in Section \ref{section_5}. The conclusion is presented in Section \ref{section_6}.

\it Notations\rm: Scalars take plain font (e.g., $L$). Bold uppercase (e.g., $\mathbf{R}$) and lowercase (e.g., $\mathbf{w}$) represent matrices and vectors, respectively. The notation $w_n$ is used for the $n$-th entry of vector $\mathbf{w}$, and $u_{i,j}$ for the $(i,j)$-th component of matrix $\mathbf{U}$. Let $\mathbb{E}(\cdot)$ represent the expectation operator. We define the following operators: $\operatorname{vec}(\cdot)$ for vectorization, $(\cdot)^*$ for complex conjugate, $(\cdot)^T$ for transpose, and $(\cdot)^H$ for Hermitian transpose. Let $\left|\cdot\right|^2$ denote the element-wise squared modulus. $\otimes$ denotes the Kronecker product. $\odot$ denotes the Hadamard (element-wise) product. $\mathrm{Diag}(\cdot)$ denotes the diagonal operator. $\delta_{i,j}$ is set to $1$ when $i=j$, and $0$ when $i\neq j$. $\delta_{c,e,s,l}$ is set to $1$ if $c=e=s=l$, and set to $0$ otherwise. $\langle\cdot\rangle_N$ denotes modulo $N$. $a\mid b$ if $b$ is a multiple of $a$, and $a\nmid b$ otherwise. $\mathrm{gcd}(x,y)$ is the greatest common divisor. $\mathrm{Card}(\cdot)$ is the number of elements in a set. $c\equiv e(\mathrm{mod~}s)$ if $s\mid(c-e)$.

\section{Signal Formulation and Sensing Metrics}\label{section_2}
\subsection{OFDM Waveform with Embedded Pilots}

This part formulates two descriptions of the pilot-embedded OFDM waveform. The first is a serial one-dimensional (1D) characterization, in which the transmitted samples are regarded as a sequence in the time domain; the second is a two-dimensional (2D) characterization that places the OFDM resource on a time-frequency lattice, with subcarriers and symbol periods forming its two axes.

\subsubsection{1D OFDM Signaling}

Consider a single transmission block consisting of $N$ samples. These samples are collected into the column vector
$\mathbf{s}= \begin{bmatrix} s_0,s_1,\ldots,s_{N-1} \end{bmatrix}^T\in\mathbb{C}^N$. Out of the $N$ entries, $L$ positions (with $L\leq N$) carry pilot symbols whose values are fixed in advance, while the remaining $(N-L)$ entries are occupied by data symbols that are mutually independent and identically distributed, each drawn from a common constellation. For the ensuing analysis, the data symbols are assumed to satisfy the statistical constraints listed below.



\begin{assumption}\label{assumption_1} \it Every data symbol $d_{s}$ is a zero-mean random variable normalized to unit power, and its pseudo-variance is zero \rm\cite{zhang2025discrete}\it:
\begin{equation}\label{eq:asp1}\mathbb{E}(d_{s})=0,\quad\mathbb{E}(\left|d_{s}\right|^2)=1,\quad\mathbb{E}(d_{s}^2)=0.\end{equation}
Moreover, the third-order moment condition is assumed:
\begin{equation}\label{eq:asp2}\mathbb{E}(d_{s}\left|d_{s}\right|^2)=\mathbb{E}(d_{s}^*\left|d_{s}\right|^2)=0.
\end{equation} 
\end{assumption}

The unit-power scaling allows a fair assessment regardless of the constellation format. \textcolor{myLightRed}{Note that} \eqref{eq:asp1} and \eqref{eq:asp2} can be satisfied by the majority of commonly used constellation formats, such as QAM and PSK. \textcolor{myLightRed}{In particular}, BPSK and 8-QAM are two \textcolor{myLightRed}{outliers} that fail the pseudo-variance condition.

To facilitate the AF analysis, \textcolor{myLightRed}{we define the kurtosis of the constellation, which is equal to} the normalized fourth central moment \textcolor{myLightRed}{and reduces to the fourth raw moment under Assumption \ref{assumption_1}:}
\begin{equation}\kappa:=\frac{\mathbb{E}\left\{\left|d_{s}-\mathbb{E}\left(d_{s}\right)\right|^{4}\right\}}{\mathbb{E}^{2}\left\{\left|d_{s}-\mathbb{E}\left(d_{s}\right)\right|^{2}\right\}}=\mathbb{E}\left\{\left|d_{s}\right|^{4}\right\}\geq1,\end{equation}
This quantity governs the structure of the AF for random signaling \cite{zhang2025discrete} \cite{11548574}. Based on its value, constellations are sorted into two families: those fulfilling $\kappa < 2$ are called sub-Gaussian, whereas those satisfying $\kappa > 2$ are called super-Gaussian. Within the former family, PSK attains the extreme $\kappa = 1$, while QAM falls into the interval $1 \leq \kappa < 2$.


Under OFDM modulation, the 1D time-domain signal is $\mathbf{x}=\mathbf{F}_N^H\mathbf{s}$, where $\mathbf{F}_N$ represents the normalized discrete Fourier transform \textcolor{myLightRed}{(DFT) matrix} of order $N$. The number of subcarriers is $N$. The set of subcarrier indices occupied by the $L$ pilots is denoted by $I=\{i_0,i_1,\ldots,i_{L-1}\}$, which defines the pilot pattern. Separating $\mathbf{s}$ according to this pattern yields a vector $\mathbf{d}$ carrying only data and a vector $\mathbf{p}$ carrying only pilots, with their entries given by:
\begin{align}\label{eq:defdi}d_i =
\begin{cases} 
0, &\text{if}\ i \in I\\
s_i, &\text{else}
\end{cases},\ \ p_i =
\begin{cases}
s_i, &\text{if}\ i \in I\\
0, &\text{else}
\end{cases}.
\end{align}
The symbol vector then admits the additive form $\mathbf{s} = \mathbf{d} + \mathbf{p}$. The pilot pattern is equivalently captured by a binary selector $\mathbf{w}$ with $w_i=1$ for $\ i \in I$ and $0$ otherwise. An illustration of the pilot pattern \textcolor{myLightPurple}{and pilot symbols} is shown in Fig. \ref{fig_initial_pattern}. 
\begin{figure}[t]
    \centering
    \includegraphics[width=0.9\linewidth]{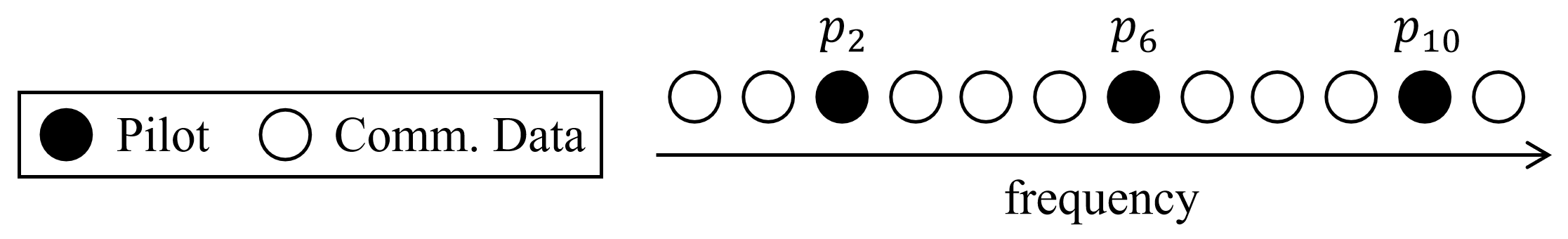}
    \caption{\textcolor{myLightBlue}{The pilot pattern and pilot symbols corresponding to $I=\{2,6,10\}$, where $N=12$, $L=3$.}}
    \label{fig_initial_pattern}
\end{figure}

To keep the comparison fair while preserving analytic tractability, the pilots are normalized to the same energy as the data symbols, that is, $\left|s_i\right|^{2}=1$ whenever $\ i \in I$. Pilots obeying this normalization are referred to as unit-amplitude pilots in what follows. \textcolor{myLightBlue}{The equal power allocation among subcarriers 
allows the AF analysis to focus on the effects of the number of pilots, pilot symbols and pilot patterns.}



\subsubsection{2D OFDM Signaling}
A two-dimensional extension of the preceding OFDM model is now considered. In this setting, the transmitted waveform is $\mathbf{X}=\mathbf{F}_N^H\mathbf{S}$, where $\mathbf{S}\in \mathbb{C}^{N\times M}$ is a symbol matrix that holds $L$ unit-amplitude pilot entries together with $(MN-L)$ data entries drawn independently and identically. The data entries are assumed to satisfy the same constellation constraints stipulated for the 1D model. The waveform $\mathbf{X}\in \mathbb{C}^{N\times M}$ extends over $N$ fast-time slots and $M$ slow-time slots.




\subsection{Evaluation Metrics for OFDM Sensing}
In this part, we adopt the discrete AFs, which were established in \cite{zhang2025discrete} for both 1D and 2D signaling, as the metrics for evaluating sensing performance.

\subsubsection{DP-AF associated with the 1D Signal}


Take a time-domain sequence comprising $N$ samples. A circular time shift by $k$ samples may be implemented via matrix multiplication. In particular, owing to the operations of CP insertion and removal, the effective shift is periodic, and we denote the corresponding operator by
\begin{equation}\mathbf{R}_{N,k}:=\begin{bmatrix}
\mathbf{0} & \mathbf{I}_k \\
\mathbf{I}_{N-k} & \mathbf{0}
\end{bmatrix},\end{equation}
where $\mathbf{I}_k$ is the identity matrix of order $k$. Likewise, a Doppler shift indexed by $q$ is realized by the frequency-domain diagonal matrix
\begin{equation}\mathbf{\Phi}_{N,q}:=\mathrm{Diag}\left(1,e^{j\frac{2\pi q}{N}},\ldots,e^{j\frac{2\pi q(N-1)}{N}}\right).\end{equation}

Suppose that the OFDM signal $\mathbf{x}$, after CP addition, passes through a multi-target channel. We characterize each of the \textcolor{myLightBlue}{${U}$ targets} by its delay index $k_\mu$, Doppler index $q_\mu$ and complex amplitude $\beta_\mu$, for $\mu=1,2,\ldots,{U}$. After removing the CP in the receiving stage, the returned echo is expressed as \cite{zhang2025discrete}
\begin{equation}\mathbf{y}=\sum_{\mu=1}^{U}\beta_{\mu}\mathbf{\Phi}_{N,q_\mu}\mathbf{R}_{N,k_\mu}\mathbf{x}+\mathbf{z},\end{equation}
where $\mathbf{z}$ is the complex Gaussian noise. By leveraging the identity $\mathbf{\Phi}_{N,q}\mathbf{R}_{N,k}=e^{-\frac{j2\pi qk}{N}}\mathbf{R}_{N,k}\mathbf{\Phi}_{N,q}$, we can express the matched-filter output as follows:
\begin{align}\label{mf_output}\nonumber&\ \tilde{y}_{k,q}^{\mathrm{MF}}=\mathbf{x}^H\mathbf{R}_{N,k}^T\mathbf{\Phi}_{N,q}^*\mathbf{y}\\&=\sum_{\mu=1}^{U}\beta_\mu e^{\frac{j2\pi k_\mu(q-q_\mu)}{N}}\mathbf{x}^H\mathbf{R}_{N,k-k_\mu}^T\mathbf{\Phi}_{N,q-q_\mu}^*\mathbf{x}+\tilde{z}_{k,q},
\end{align}
where $\tilde{z}_{k,q}=\mathbf{x}^H\mathbf{R}_{N,k}^T\mathbf{\Phi}_{N,q}^*\mathbf{z}$. The signal-dependent term presented above is referred to as the DP-AF:
\begin{align}\label{dpaf_def}\mathcal{X}^{\mathrm{DP}}(k,q):=\mathbf{x}^H\mathbf{R}_{N,k}^T\mathbf{\Phi}_{N,q}^*\mathbf{x},\quad(k,q)\in\mathbb{Z}_N^2,\end{align}
in which $k$ and $q$ represent the delay index and the Doppler index, respectively\textcolor{myLightRed}{, and} $\mathbb{Z}_N$
is the cyclic group of order $N$, realized as the set $\{0,1,..., N-1\}$ under addition modulo $N$.

The matched-filter output in \eqref{mf_output} is a weighted sum of  time-frequency shifted DP-AFs. For robust multi-target detection, the squared matched-filter response, denoted as $|\tilde{y}_{k,q}^{\mathrm{MF}}|^2$, is expected to exhibit prominent and well-separated peaks precisely at the coordinates $(k,q)=(k_\mu,q_\mu)$, while residual sidelobes should remain consistently low elsewhere. This behavior critically relies upon the DP-AF.
The sidelobe level is therefore quantified as \cite{zhang2025discrete}:
\begin{align}\label{define_squared_DPAF}
|\mathcal{X}^{\mathrm{DP}}(k,q)|^2 
 & =|\mathbf{x}^H\mathbf{\Phi}_{N,q}\mathbf{R}_{N,k}\mathbf{x}|^2,\quad(k,q)\neq(0,0),
\end{align}
whereas the mainlobe level equals $|\mathcal{X}^{\mathrm{DP}}(0,0)|^2=|\mathbf{x}^H\mathbf{x}|^2$.

To facilitate the calculation, the property of central symmetry can be used as follows:
\begin{align}\label{central_symmetry}
|\mathcal{X}^{\mathrm{DP}}(k,q)|^2=|\mathcal{X}^{\mathrm{DP}}({\langle N-k\rangle_N},{\langle N-q\rangle_N})|^2.
\end{align} 

\subsubsection{FST-AF associated with the 2D Signal}
Let the stacked OFDM vector $\mathbf{x}=\operatorname{vec}\left(\mathbf{X}\right)$ of dimension $MN$ be sent. Provided the Doppler shift is sufficiently small, the phase rotation induced by a target is approximately uniform over each $N$-sample fast-time block. Under this premise, the vectorized waveform can be rearranged into $M$ columns, giving rise to the 2D fast-slow-time (FST) arrangement $\mathbf{X}\in\mathbb{C}^{N\times M}$. Moreover, a CP is prepended per slow-time data segment at the transmitter so as to cover the largest expected delay. With several targets in the scene, the received signal after discarding the CP is captured by \cite{zhang2025discrete}:
\begin{equation}\mathbf{y}\approx\sum_{\mu=1}^{U}\beta_\mu\left(\mathbf{\Phi}_{M,q_\mu}\otimes\mathbf{R}_{N,k_\mu}\right)\mathbf{x}+\mathbf{z}.\end{equation}
After matched filtering, the output is:
\begin{align}
 & \nonumber\ \tilde{y}_{k,q}^{\mathrm{MF}}=\mathbf{x}^{H}\left(\mathbf{\Phi}_{M,q}^{*}\otimes\mathbf{R}_{N,k}^{T}\right)\mathbf{y} \\
 & \approx\sum_{\mu=1}^{{U}}\beta_{\mu}\mathbf{x}^H\left(\mathbf{\Phi}_{M,q-q_\mu}^*\otimes\mathbf{R}_{N,k-k_\mu}^T\right)\mathbf{x}+\tilde{z}_{k,q},
\end{align}
with $\tilde{z}_{k,q}=\mathbf{x}^H\left(\mathbf{\Phi}_{M,q}^*\otimes\mathbf{R}_{N,k}^T\right)\mathbf{z}$. The signal-dependent term appearing above is termed the FST-AF:
\begin{equation}\mathcal{X}^{\mathrm{FST}}\left(k,q\right):=\mathbf{x}^H\left(\mathbf{\Phi}_{M,q}^*\otimes\mathbf{R}_{N,k}^T\right)\mathbf{x},\end{equation}
where $k\in\mathbb{Z}_N$ and $q\in\mathbb{Z}_M$. By leveraging the identity $\mathbf{R}_{N,k} = \mathbf{F}_N^H \mathbf{\Phi}_{N,k}^* \mathbf{F}_N$, the FST-AF admits a compact representation \cite{zhang2025discrete}:
\begin{align}\mathcal{X}^\mathrm{FST}&=\sqrt{MN}\mathbf{F}_N^H\left|\mathbf{F}_N\mathbf{X}\right|^2\mathbf{F}_M.
\end{align}

Specializing to the 2D OFDM waveform, the above expression reduces to
\begin{align}\label{define_FSTAF}\mathcal{X}^\mathrm{FST}&=\sqrt{MN}\mathbf{F}_N^H\left|\mathbf{S}\right|^2\mathbf{F}_M.
\end{align}

\section{Statistical AF Analysis of Pilot-Embedded OFDM Signals}\label{section_3}
The transmitted random data symbols render the DP-AF and the FST-AF as random functions, thereby making the study of their statistical properties essential. For an AF $\mathcal{X}(k,q)$ and $(k,q)\neq(0,0)$, the expected sidelobe level (ESL) represents the expected local sidelobe power, expressed as $\mathbb{E}(|\mathcal{X}(k,q)|^2)$. The expected integrated sidelobe level (EISL) represents the total expected energy contained in all sidelobes, expressed as $\Gamma(\mathcal{X})=\sum_{(k,q)\neq(0,0)}\mathbb{E}(|\mathcal{X}(k,q)|^2)$.

In the rest of this section, we derive
closed-form ESL and EISL for the DP-AF and FST-AF of pilot-embedded OFDM signals, and discuss how these metrics depend on the number of pilots, pilot symbols, pilot patterns, and \textcolor{myLightPurple}{the 
kurtosis.}

\subsection{Analytical Expression for the Expected Squared DP-AF}
In the following theorem, we express the expectation of \eqref{define_squared_DPAF} in an analytical form.

\begin{thm}\label{theorem_1}
For \textcolor{myLightBlue}{the OFDM signals} embedding unit-amplitude pilots, the expected squared DP-AF is \rm
\textcolor{myLightPurple}{\begin{equation}
    \begin{aligned}\label{dp_af_theorem}\ &\ \mathbb{E}\left(|\mathcal{X}^{\mathrm{DP}}(k,q)|^2\right)=\begin{cases} N^2\delta_{k,0}+(\kappa-1)(N-L),  & \text{if}\ q=0,\\
\psi(k,q),  &\text{if}\ q\neq0,\end{cases}\end{aligned} \end{equation} }
\it where $\psi(k,q)$ depends on the pilot symbols and pilot patterns: 
\begin{align}\label{f_g_h_relation}
&\nonumber\ \psi(k,q) 
\\&\nonumber=\left|\sum_{n=0}^{N-1}p_np_{\langle n-q\rangle_N}^*\cdot e^{j2\pi nk/N}\right|^2+N-\sum_{n=0}^{N-1}w_nw_{\langle n-q\rangle_N}
\\&\nonumber=\left|\mathbf{p}^H\mathbf{F}_N\mathbf{\Phi}_{N,q}\mathbf{R}_{N,k}\mathbf{F}_N^H\mathbf{p}\right|^2+N-\mathbf{w}^H\mathbf{F}_N\mathbf{\Phi}_{N,q}\mathbf{F}_N^H\mathbf{w}
\\&= \textcolor{myLightRed}{g(k,q) + N - h(q)}.\end{align}
\end{thm}

\begin{proof}
    \textcolor{myLightRed}{See Appendix \ref{proof_of_theorem_1}. }
\end{proof}

For brevity, we denote $\left|\mathbf{p}^H\mathbf{F}_N\mathbf{\Phi}_{N,q}\mathbf{R}_{N,k}\mathbf{F}_N^H\mathbf{p}\right|^2$ as $g(k,q)$ and $\mathbf{w}^H\mathbf{F}_N\mathbf{\Phi}_{N,q}\mathbf{F}_N^H\mathbf{w}$ as $h(q)$.
\textcolor{myLightBlue}{The function $g(k,q)$ is determined jointly by the pilot symbols and their pattern, whereas $h(k,q)$ depends solely on the pilot pattern.
It can be observed that $g(k,q)$ is the DP-AF of the pilot component $\mathbf{F}_N^H\mathbf{p}$, where $\mathbf{p}=\mathbf{s}\odot \mathbf{w}$ is formed by inserting $(N-L)$ zeros into the pilot sequence, 
with the locations of pilots dictated by the pilot pattern mask $\mathbf{w}$.}
\textcolor{myLightRed}{The function} $h(q)$ is the cyclic autocorrelation of the pilot pattern mask $\mathbf{w}$, where the cyclic shift amount corresponds to the Doppler index $q$.


The ESL 
\textcolor{myLightPurple}{for} $q=0$ is $(\kappa-1)(N-L)$, which depends on the signal length, the number of pilots and the constellation kurtosis. The ESL for $q=0$ equals zero when $L=N$ or the kurtosis $\kappa = 1$, for example, in PSK constellations. 
The ESL when $q\neq0$ is $\psi(k,q)$, which varies with the pilot symbols and their patterns, and is independent of the kurtosis $\kappa$. 

\subsection{Sum of the Expected Sidelobe Levels for the DP-AF}
Next we analyze the invariance of the sum of expected sidelobe levels for the DP-AF. From \cite{wu2026ambiguity}, the EISL  is:
\begin{align}
\Gamma(\mathcal{X}^{\mathrm{DP}}) =\left[N^2+(\kappa-1)(N-L)\right](N-1).
\end{align}

Hence, for a given signal length, $\Gamma(\mathcal{X}^{\mathrm{DP}})$ is determined solely by the pilot count $L$ and the constellation kurtosis $\kappa$. It follows from \eqref{dp_af_theorem} that the sum of the sidelobes \textcolor{myLightPurple}{over delay} at the zero-Doppler cut is \textcolor{myLightRed}{$(\kappa-1)(N-1)(N-L)$}. For a non-zero Doppler index, the following result can be \textcolor{myLightPurple}{derived:} 

\begin{corollary}
Under any given pilot pattern and unit-amplitude pilot symbols, the sum of expected sidelobe levels of DP-AF for any $q\neq0$ remains constant:
\begin{align}\label{sum_of_sidelobe}\sum_{k=0}^{N-1}\mathbb{E}\left(|\mathcal{X}^{\mathrm{DP}}(k,q)|^2\right)=
N^2. \end{align} \rm
\end{corollary}

\begin{proof}
\textcolor{myLightRed}{Using $|p_n|=|p_n|^2=w_n$, 
$\psi(k,q)$ can be expanded as}
\begin{align}\label{DPAF_expand}
\sum_{n\neq n'}p_np_{\langle n-q\rangle_N}^* e^{j2\pi (n-n')k/N}p_{n'}^*p_{\langle n'-q\rangle_N}+N.
\end{align}
Plugging $\sum_{k=0}^{N-1}e^{j2\pi(n- n')k/N}=0$ into the summation of \eqref{DPAF_expand} leads to \eqref{sum_of_sidelobe}.
\end{proof}


Note that \eqref{sum_of_sidelobe} is independent of the \textcolor{myLightPurple}{pilots,}  
revealing a 
theoretical 
limitation
\textcolor{myLightPurple}{that} we cannot optimize the sum of the expected sidelobe levels for any $q\neq0$ solely by adjusting the pilot 
design.

\subsection{Impact of Pilot Patterns on the Expected Squared DP-AF}

To explore the practical design space of pilot patterns, we turn to the following corollary: \newpage


\begin{corollary}
\textcolor{myLightBlue}{For OFDM signals}, if the unit-amplitude pilot pattern is cyclically shifted by $c$ positions, yielding a shifted pattern $p_{\langle n-c\rangle_N}$ and $w_{\langle n-c\rangle_N}$, the expected squared DP-AF remains invariant.
\end{corollary}

\begin{proof}
\textcolor{myLightRed}{The expected squared DP-AF when $q\neq 0$ after shifting} the positions of pilots can be given as
\begin{align}
&\ \nonumber \left|\sum_{n=0}^{N-1}p_{\langle n-c\rangle_N}p_{\langle {\langle n-q\rangle_N}-c\rangle_N}^{*}\cdot e^{j2\pi nk/N}\right|^2 + N\\&\nonumber-\sum_{n=0}^{N-1}w_{\langle n-c\rangle_N}w_{\langle {\langle n-q\rangle_N}-c\rangle_N}\\&\nonumber=\left|\sum_{n=0}^{N-1}p_{\langle n-c\rangle_N}p_{\langle {\langle n-c\rangle_N}-q\rangle_N}^{*}\cdot e^{j2\pi {\langle n-c\rangle_N}k/N}\right|^2 + N\\&-\sum_{n=0}^{N-1}w_{\langle n-c\rangle_N}w_{\langle {\langle n-c\rangle_N}-q\rangle_N},
\end{align}\rm
which equals that in the unshifted case.
\end{proof}

This result provides flexibility in pilot placement without altering the sidelobe behavior. To further gain insights into the impact of pilot patterns, we derive the attainable bounds:


\begin{proposition}
For the OFDM signals where unit-amplitude pilots are embedded, under arbitrary pilot design, when $q\neq0$, the lower bound of the expected squared DP-AF equals $(N-L)$:
\begin{align}\label{lower_bound}\mathbb{E}\left(|\mathcal{X}^{\mathrm{DP}}(k,q)|^2\right)\geq N-h(q)\geq N-L,\end{align}
and the upper bound when $q\neq0$ is $(N+L^2-L)$:
\begin{align}\label{upper_bound}\nonumber\mathbb{E}\left(|\mathcal{X}^{\mathrm{DP}}(k,q)|^2\right)&\leq N+ h^2(q)-h(q)\\&\leq N+L^2-L.\end{align}
\end{proposition}

\begin{proof}
Using $|p_n|=|p_n|^2=w_n$, it can be proved that \begin{align}g(k,q)&\leq\left(\sum_{n=0}^{N-1}\left|p_np_{\langle n-q\rangle_N}^*\cdot e^{j2\pi nk/N}\right|\right)^2
=h^2(q).\end{align}
Furthermore, combining this with $h(q)\leq L$ yields the upper bound. The lower bound follows from $g(k,q)\geq0$.
\end{proof}

Recalling that the ESL at $q=0$ is $(\kappa-1)(N-L)$, we note that if $L<N$, for the \textcolor{myLightPurple}{super-Gaussian constellation}, this ESL lies above the lower bound, whereas for the \textcolor{myLightPurple}{sub-Gaussian constellation}, 
\textcolor{myLightPurple}{this ESL} 
attains the global minimum.

The upper bound is achieved only when $g(k,q)=L^2$ and $h(q)=L$, while the lower bound is achieved only when $g(k,q)=0$ and $h(q)=L$. \textcolor{myLightRed}{Note that} $h(q)=L$ means the pilot pattern mask $\mathbf{w}$ is cyclic-shift invariant under Doppler index $q$. For $q\neq0$, define the set of Doppler indices for which it is possible to reach the upper and lower bounds as
\begin{align}
    \mathcal{D}_h:=\left\{q\in\mathbb{Z}_N,q\neq0|h(q)=L\right\},
\end{align}
the set of delay-Doppler index pairs for which the lower bound is attained as 
\begin{align}
    \mathcal{L}_{g,h}:=\left\{(k,q)\in\mathbb{Z}_N,q\neq0|q\in\mathcal{D}_h,\ g(k,q)=0\right\},
\end{align}
and the set of delay-Doppler index pairs for which the upper bound is attained as 
\begin{align}
    \mathcal{U}_{g,h}:=\left\{(k,q)\in\mathbb{Z}_N,q\neq0|q\in\mathcal{D}_h,\ g(k,q)=L^2\right\}.
\end{align}

We first reveal that when there exist sidelobes reaching the upper or lower bounds, the pilot pattern must be periodic. Furthermore, when the number of Doppler indices that can generate sidelobes reaching the upper or lower bounds is maximized, the pilot pattern is equally spaced.

\begin{proposition} \label{proposition_2}
\textcolor{myLightRed}{The set $\mathcal{D}_h\neq\varnothing$ if and only if $\mathrm{gcd}(N,L)>1$} and $\mathbf{w}$ has period $T_\mathbf{w}=N/\lambda_{N,L}$, i.e., $w_n=w_{\langle n-T_\mathbf{w}\rangle_N},\ \forall n\in \mathbb{Z}_N$, where $\lambda_{N,L}>1$ is a common divisor of $N$ and $L$. 

In this case, $\mathcal{D}_h$ takes the following form:
\begin{align}\label{D_h_form}
\mathcal{D}_h=\{iT_\mathbf{w}|i=1,...,\lambda_{N,L}-1\}.\end{align}

It follows that $\mathrm{Card}\left(\mathcal{D}_h\right)\leq L-1$; the equality is satisfied only when $L\mid N$ and the pilot pattern is equally spaced with a period of $N/L$.
\end{proposition}


\begin{proof}
    See Appendix \ref{proof_of_proposition_2}.
\end{proof}

Based on the above result, we \textcolor{myLightRed}{show in the following} that, when the equally spaced pilot pattern is used and chirp sequences are adopted as pilot symbols, the numbers of sidelobes that attain the upper and lower bounds are maximized simultaneously.

\begin{proposition}\label{proposition_3}
\textcolor{myLightRed}{For the number of sidelobes that attain the} lower bound, we have:
\begin{align}\label{LB_max}
\mathrm{Card}\left(\mathcal{L}_{g,h}\right)\leq \frac{N}{L}(L-1)\cdot\mathrm{Card}\left(\mathcal{D}_h\right)\leq\frac{N}{L}(L-1)^2,
\end{align}
and for the number of sidelobes that attain the upper bound, we have:
\begin{align}\label{UB_max}
\mathrm{Card}\left(\mathcal{U}_{g,h}\right)\leq \frac{N}{L}\cdot\mathrm{Card}\left(\mathcal{D}_h\right)\leq\frac{N}{L}(L-1).
\end{align}
The two inequalities in \eqref{LB_max} are tight only when the pilot pattern is equally spaced with a period of $N/L\in \mathbb{Z}$ and the pilot symbols take the following form of chirp sequence: \rm
\begin{equation}\label{chirp_main}
    p_{l\frac{N}{L}}=\begin{cases}
        A\cdot e^{\left[-j\frac{\pi u}{L}l(l+1)\right]}\cdot e^{\left(j\frac{2\pi\nu}{L}l\right)},&\text{if $L$ is odd,}\\ 
        A\cdot e^{\left(-j\frac{\pi u}{L}l^2\right)}\cdot e^{\left(j\frac{2\pi\nu}{L}l\right)},&\text{if $L$ is even,}
    \end{cases}
\end{equation}
\it where $l=0,...,L-1$, $A\in\mathbb{C}$, $|A|=1$, and $u,\nu \in \mathbb{Z}$. Similarly, those in \eqref{UB_max} are tight only under the same conditions. The corresponding expected squared DP-AF is \rm
\begin{equation}\begin{aligned}\label{chirp_dp_af}\ &\ \mathbb{E}\left(|\mathcal{X}^{\mathrm{DP}}(k,q)|^2\right)=\\&\begin{cases} N^2\delta_{k,0}+(\kappa-1)(N-L), & \text{if}\ q=0,\\
N+L^2-L,\  & \text{if}\ q=m\frac{N}{L},\ L\mid(k-um),
\\N-L,\   & \text{if}\ q=m\frac{N}{L},\ L\nmid(k-um),\\
N,  &\text{else},\end{cases}\end{aligned}\end{equation} \it 
where $m = 1,2, ..., L-1$. \eqref{chirp_dp_af} is independent of $A$ and $\nu$. Notably, when $A=1$ and $\nu=0$, the chirp sequence in \eqref{chirp_main} reduces to the ZC sequence. 
\end{proposition} 

\begin{proof}
    See Appendix \ref{proof_of_proposition_3}.
\end{proof}

Under the above pilot design, the $N$ sidelobes corresponding to any $q\in \mathcal{D}_h=\{i(N/L)|i=1,...,L-1\}$ are composed of $N/L$ sidelobes that attain the upper bound, and $(N-N/L)$ sidelobes that attain the lower bound. 

From \eqref{sum_of_sidelobe}, no pilot design can minimize the sum of sidelobes of the DP-AF for a given number of pilots. A natural question is whether, for a given number of pilots, one can design the pilot symbols and their pattern to minimize the sidelobe level in a specified delay-Doppler region. From \eqref{chirp_dp_af}, such minimization is achievable with equally spaced chirp 
\textcolor{myLightPurple}{pilots}, provided that the area of interest is restricted to certain 1D subsets, for instance, certain adjacent sidelobes corresponding to a non-zero Doppler index $q=m(N/L)$. For realistic deployments, though, one typically prefers to confine the minimal sidelobes to a connected two-dimensional area made up of neighboring delay-Doppler bins, which leads to the idea of the low-ambiguity zone \cite{ye2022low}. Pilot designs constructed in this way can therefore be regarded as locally optimal within that prescribed 
\textcolor{myLightPurple}{2D} 
region. As will be shown, however, such locally optimal pilot designs generally do not exist.



\begin{corollary}
   
For any $q\in\mathbb{Z}_N$, $q$ and $\langle q+1\rangle_N$ both belong to $\mathcal{D}_h\cup\{0\}$ only when $L=N$.

\end{corollary}

\begin{proof}
If $q,\langle q+1\rangle_N\in\mathcal{D}_h\cup\{0\}$, from \eqref{D_h_form} it can be deduced that $T_\mathbf{w}=1$ and $\lambda_{N,L}=N$, thus $L=N$. 
\end{proof}

Therefore, when $L<N$ there does not exist a $2\times2$ delay-Doppler region, or a low-ambiguity zone that achieves the minimum sidelobes.

\subsection{Statistical Characterization of the FST-AF}
The analytical expression for the expected squared FST-AF is established in the subsequent theorem.

\begin{thm} \label{theorem_2}
For the OFDM signals in which unit-amplitude pilots are embedded, the expected squared FST-AF \textcolor{myLightPurple}{is}
\begin{equation}
    \begin{aligned}\label{eq:expect_fst}&\ \textcolor{myLightPurple}{\mathbb{E}\left(|\mathcal{X}^{\mathrm{FST}}(k,q)|^2\right)=M^2N^2\delta_{k,0}\delta_{q,0}+(\kappa-1)(MN-L).}\end{aligned}
\end{equation}
\end{thm}

\begin{proof}
    \textcolor{myLightRed}{See Appendix \ref{proof_of_theorem_2}.}
\end{proof}

The ESL associated with the FST-AF is
$(\kappa-1)(MN-L)$, which
equals zero when $L$ equals the total number of time-frequency resources $MN$, or the kurtosis $\kappa = 1$.

An important observation within the FST-AF model is that the expected sidelobe levels are identical and independent of the pilot patterns and pilot symbols. This result suggests that the small-Doppler assumption in 2D OFDM signal models may fail to properly capture the effect of pilots on realistic sensing performance. The EISL of FST-AF is expressed as:
\begin{align}
\Gamma(\mathcal{X}^{\mathrm{FST}}) & =(\kappa-1)(MN-1)(MN-L),
\end{align}
which is dictated by the pilot count, the constellation kurtosis, and the number of time-frequency resources.

\section{Case Study on Using ZC Sequences as Pilots}\label{section_4}

The ZC sequences are commonly used as pilot symbols, including the sounding reference signals \textcolor{myLightPurple}{(SRS)} 
\cite{cha20255g} 
and the preamble of the physical random access channel \textcolor{myLightPurple}{(PRACH)} 
\cite{tosun2026coordinated} 
in 5G new radio (NR) systems. \newpage

The ideal correlation properties \textcolor{myLightPurple}{of ZC sequences} make them particularly attractive for use in channel estimation, synchronization and \textcolor{myLightBlue}{radar probing tasks}. In communication protocols, odd-length ZC sequences are typically employed, and take the form
\textcolor{myLightPurple}{$z_l = e^{-j\frac{\pi u }{L_{z}}l(l+1)},\ 0 \leq l \leq L_{z}-1,$}
where $L_{z}$ is the sequence length, and $u$ is the root index, usually chosen \textcolor{myLightRed}{co-prime} to $L_{z}$ \cite{pitaval2020overcoming}. Assume the number of \textcolor{myLightPurple}{pilots} 
equals the full sequence length, i.e., $L = L_{z}$. The subcarrier assignment is given by $z_l  = p_{i_l}$ for all $l$, where the $L$ subcarrier indices 
designated 
for pilots are 
first 
determined as $i_0 < i_1 <...< i_{L-1}$. This ordered arrangement aligns the 
ZC 
sequence with the ascending pilot subcarrier indices \cite{agli2025zadoff}, thereby simplifying the extraction and ordering of received pilots.

\begin{figure}[t]
    \centering
    \includegraphics[width=0.9\linewidth]{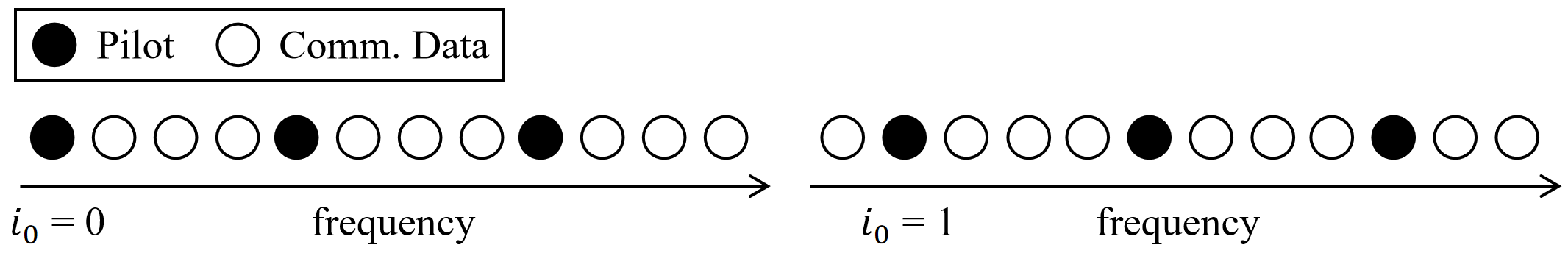}
    \caption{The pilot pattern corresponding to equally spaced placement, where $N=12$, $L=3$.}
    \label{fig_equi_pattern}
\end{figure}

\begin{figure}[t] 
    \centering
    \begin{subfigure}[t]{.49\linewidth}
    \centering
    \includegraphics[width=1\linewidth]{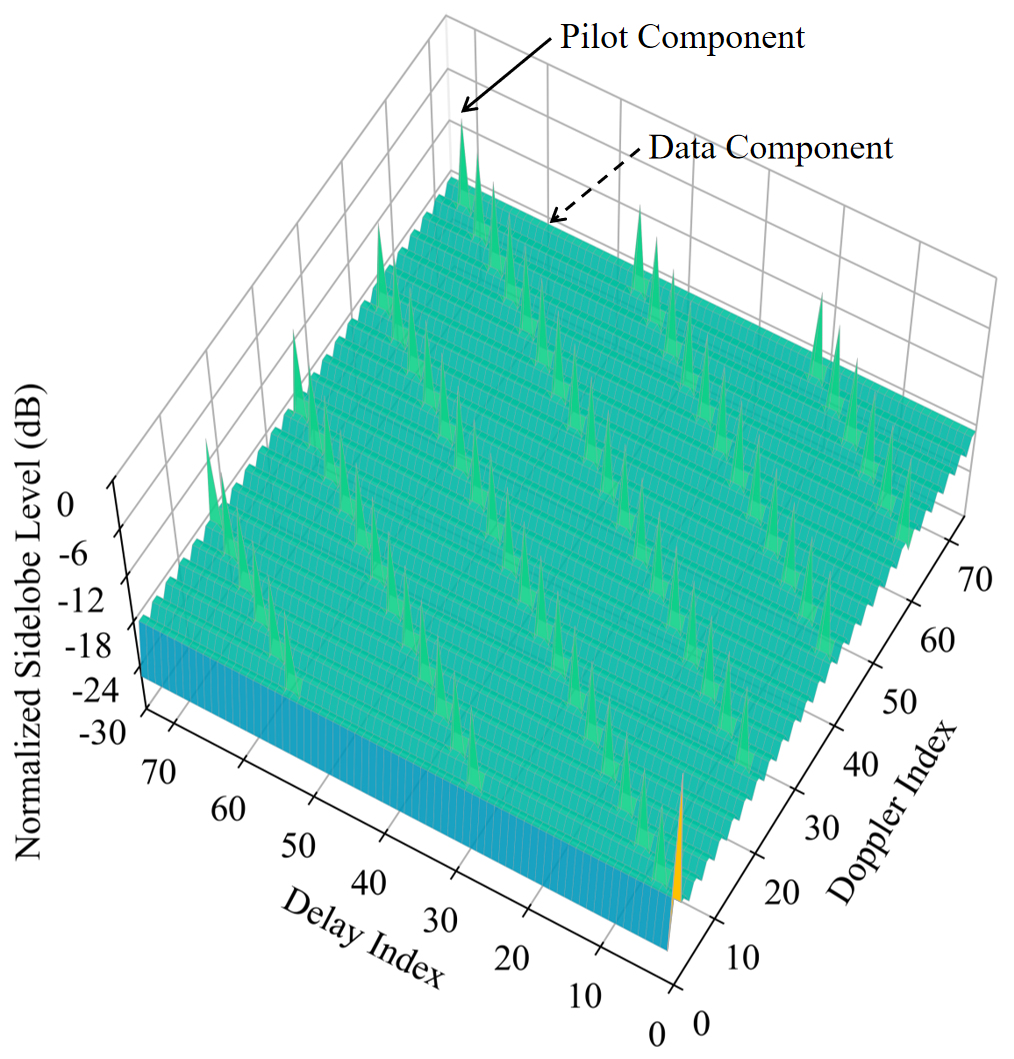}
    \caption{\textcolor{myLightBlue}{The case when $L=25$.}}
    \label{1_1_comb_theo_dpaf_3d_N_75_L_25_u_4}
    \end{subfigure}
    \centering
    \begin{subfigure}[t]{.49\linewidth}
    \centering
    \includegraphics[width=1\linewidth]{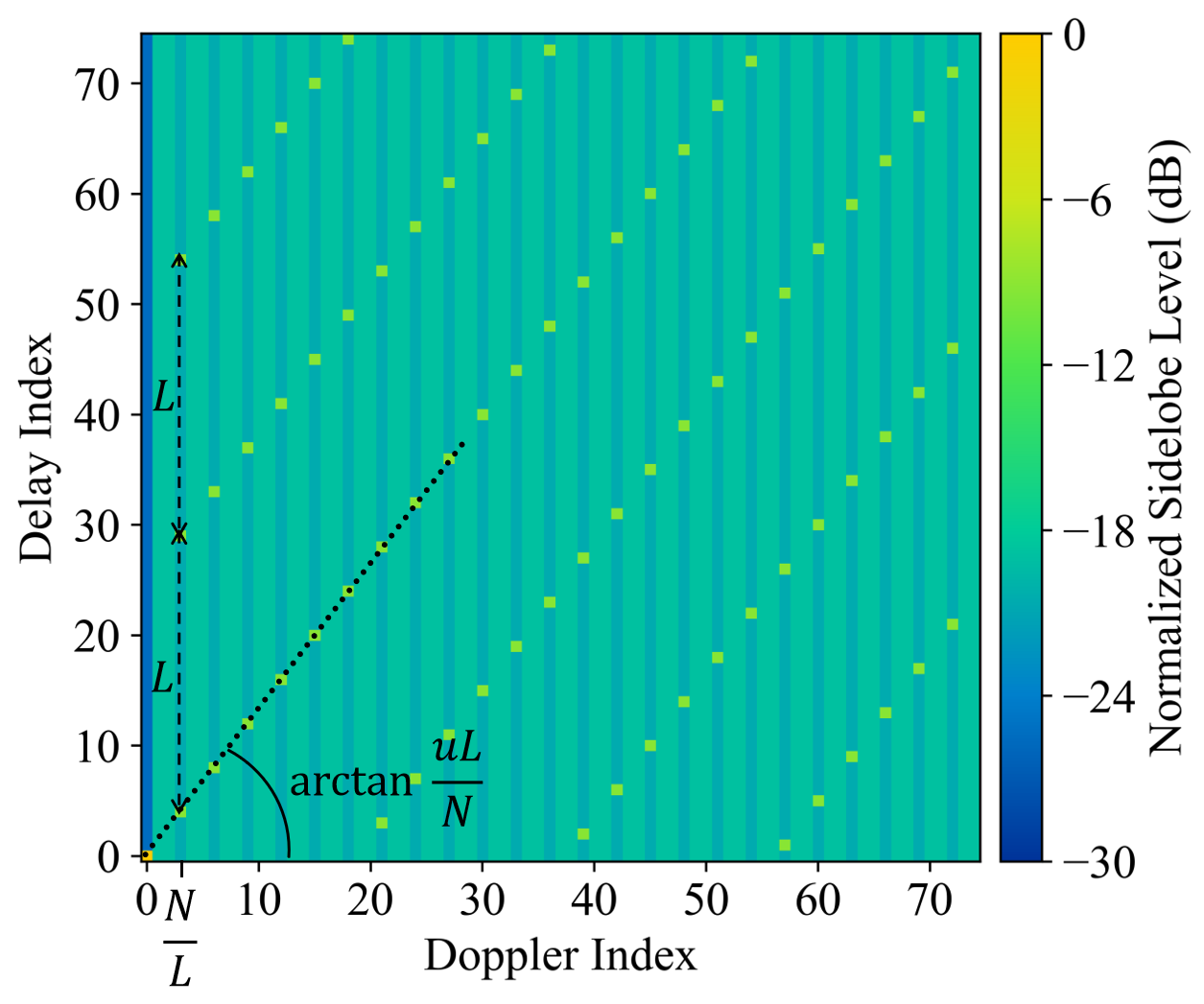}
    \caption{2D view of (a).}
    \label{fig:comb_theo_dpaf_3d_and_2d}
    \end{subfigure}
    \centering
    \begin{subfigure}[t]{.49\linewidth}
    \centering
    \includegraphics[width=1\linewidth]{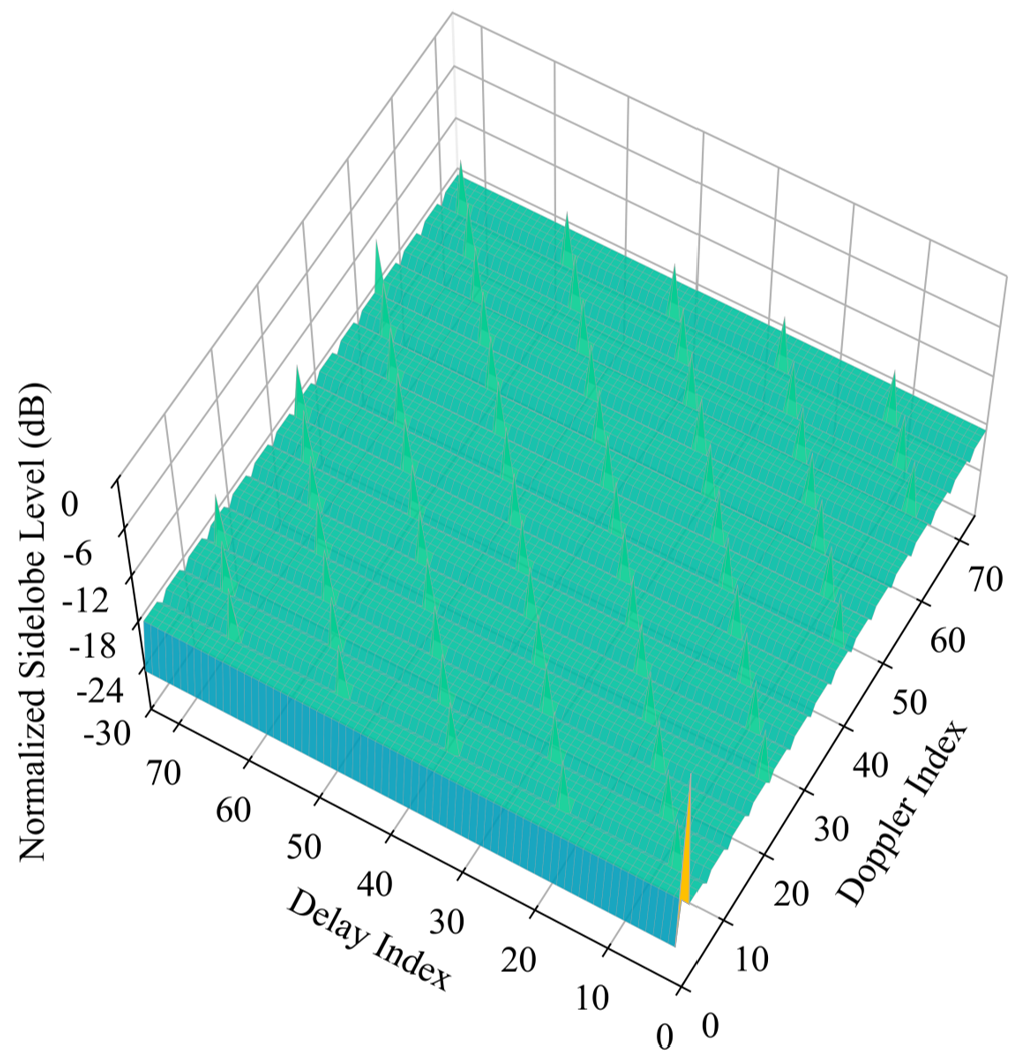}
    \caption{The case when $L=15$.}
    \end{subfigure}
    \centering
    \begin{subfigure}[t]{.49\linewidth}
    \centering
    \includegraphics[width=1\linewidth]{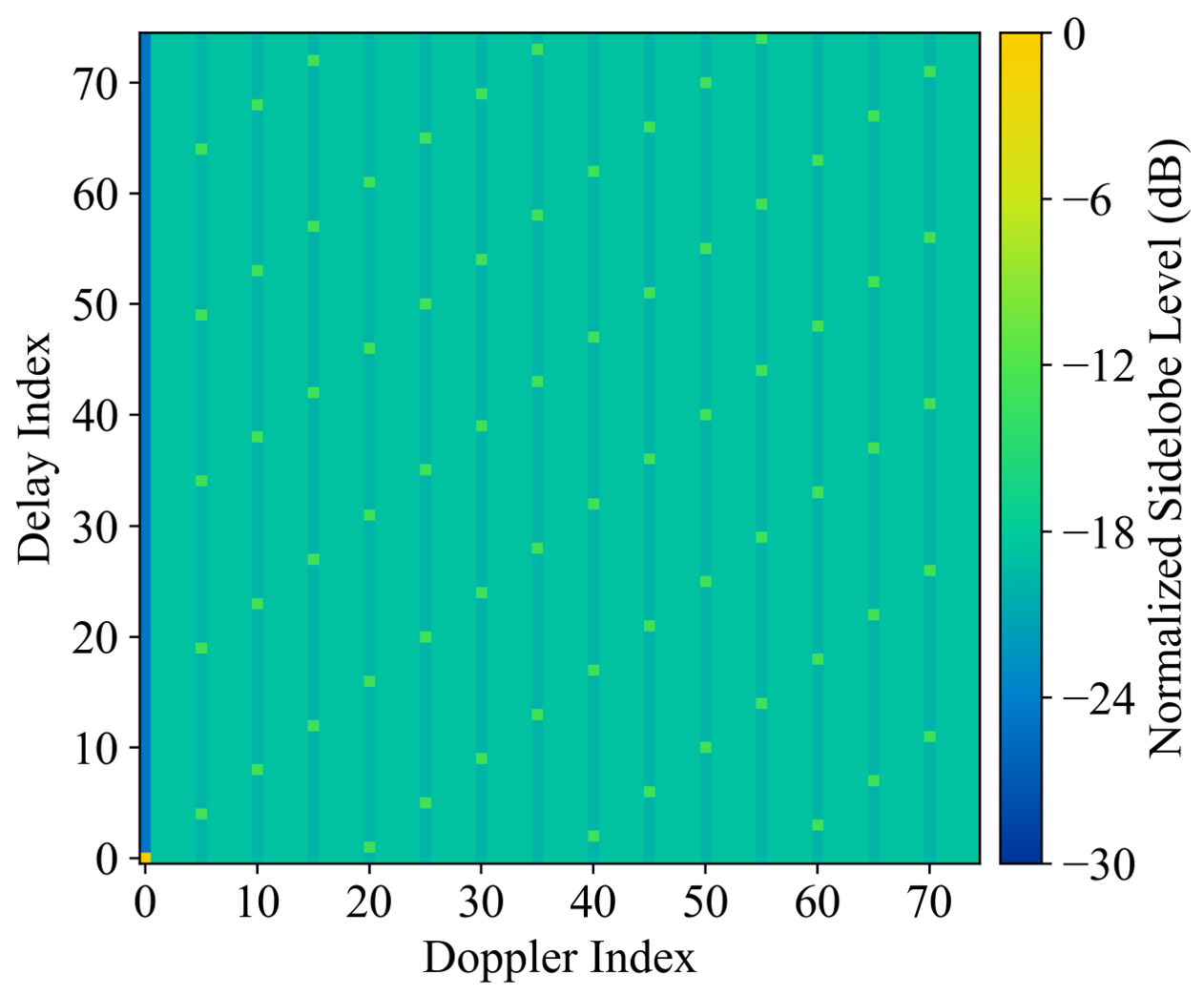}
    \caption{2D view of (c).}
    \end{subfigure}
    \caption{Theoretical expected squared DP-AF of the OFDM signal with equally spaced ZC sequences, where $N=75$, $u=4$, under 16-QAM constellation.}
    \label{fig:all_comb_theo_dpaf_3d_and_2d}
\end{figure}


In this section, we analyze the expected squared DP-AF of OFDM signals using odd-length ZC sequences as pilots under two pilot patterns: equally spaced placement and contiguous placement. \textcolor{myLightBlue}{The choice of these two patterns reflects their counterparts in practical OFDM systems. Equally spaced placement corresponds to the comb-type pilot pattern widely used for channel estimation, where uniformly distributed pilots, such as the
demodulation reference signals (DMRS) \cite{11322556},
allow the receiver to interpolate the frequency-domain channel response. 
Contiguous placement corresponds to the localized pilot arrangement used for synchronization and initial acquisition, where a cluster of adjacent known symbols, such as the secondary synchronization signals (SSS) \cite{10502156}, serves as the timing or acquisition reference. 
Analyzing both patterns thus covers the main roles of pilots.} 



\begin{figure}[t]
    \centering
    \includegraphics[width=0.9\linewidth]{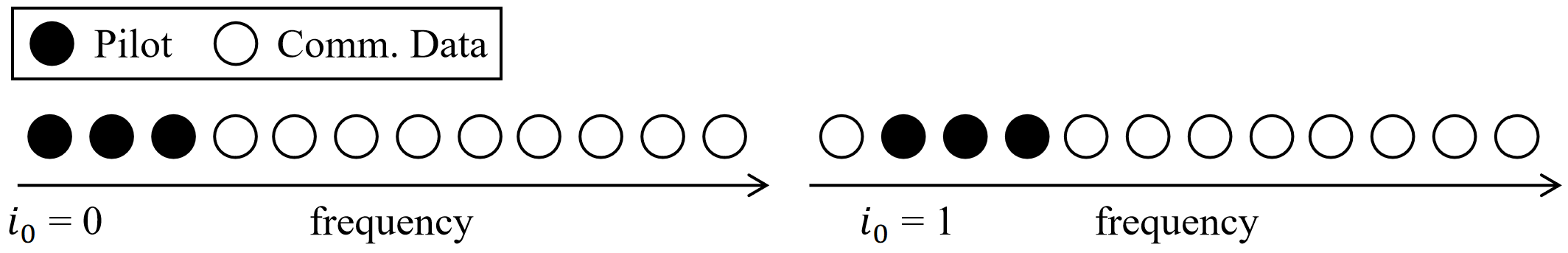}
    \caption{The pilot pattern corresponding to contiguous placement, where $N=12$, $L=3$.}
    \label{fig_cont_pattern}
\end{figure}

\subsection{Equally Spaced Placement as the Pilot Pattern}

In this subsection, the subcarrier indices for pilots are separated by a fixed \textcolor{myLightPurple}{interval.} 
We consider the case where $N$ is a multiple of $L$, and set the spacing to \textcolor{myLightPurple}{$N/L$} \cite{chen2006peak}. The pilot indices are given by $i_l= l(N/L)+i_0,\ l=0,...,L-1$. 
\textcolor{myLightPurple}{An example} illustrating the pilot pattern is shown in Fig. \ref{fig_equi_pattern}.

The corresponding expected squared DP-AF is given in \eqref{chirp_dp_af}. On $(N-L)$ Doppler bins, the expected sidelobe levels exhibit a flat baseline equal to $N$, indicating a noise-like ambiguity response contributed entirely by the random data component. In contrast, on the Doppler slices indexed by $m(N/L)$, the AF develops a distinct comb-shaped structure along the delay dimension. The delay-domain profile alternates between the upper bound and the lower bound according to the congruence relation $L\mid(k-um)$ or \textcolor{myLightRed}{$k\equiv um(\mathrm{mod~}L)$}. As shown in Fig. \ref{fig:comb_theo_dpaf_3d_and_2d}, the peaks repeat with period $L$ over delay and shift linearly with the Doppler parameter $m$, tracing a family of parallel lines with slope $uL/N$ in the discrete delay-Doppler plane.

The appearance of peaks exclusively at $q=m(N/L)$ can be explained by the equal spacing of the ZC sequence in the frequency domain. A Doppler shift of $m(N/L)$ exactly aligns the embedded ZC sequence with its shifted replica. The linear congruence condition $k\equiv um(\mathrm{mod~}L)$ is the discrete counterpart of the range-Doppler coupling typical of linear chirps \cite{berggren2022joint}, where the delay of the correlation peak is proportional to the Doppler shift with slope determined by the chirp rate (here controlled by the root index $u$). Where this condition fails, the perfect autocorrelation property of the ZC sequence translates into a suppression below the data floor, creating a notched delay profile. By designing $u$, $L$, and $N/L$, one can precisely control the positions and the prominence of these sensing peaks.

Furthermore, note that on the zero-delay cut, the condition $k\equiv um(\mathrm{mod~}L)$ cannot be satisfied under the conventional assumption that $u$ and $L$ are co-prime \cite{pitaval2020overcoming}. Therefore
\begin{equation}\begin{aligned}\label{zc_dp_af_zero_delay}\ &\ \mathbb{E}\left(|\mathcal{X}^{\mathrm{DP}}(0,q)|^2\right)\\&=\begin{cases} N^2+(\kappa-1)(N-L), & \text{if}\ q=0,\\
N-L,\   & \text{if}\ q=m\frac{N}{L},\ m>0,\\
N,  &\text{else}.\end{cases}\end{aligned}\end{equation}

\subsection{Contiguous Placement as the Pilot Pattern}
In this subsection, the pilots are placed contiguously over a block of adjacent \textcolor{myLightPurple}{subcarriers} \cite{syrjala2019pilot}. The occupied indices are set as $i_l= i_0+l,\ l=0,...,L-1$. 
\textcolor{myLightPurple}{An example} of the pattern is shown in Fig. \ref{fig_cont_pattern}. The corresponding expected squared DP-AF is expressed as \textcolor{myLightRed}{follows}.


\begin{proposition}\label{proposition_4}
\textcolor{myLightRed}{When  $L\leq N/2$, the expected squared DP-}AF of the OFDM signals containing contiguously placed odd-length ZC sequences is \rm
\begin{align}\label{block_dp_af_l_xiao}
&\ \nonumber\mathbb{E}\left(|\mathcal{X}^{\mathrm{DP}}(k,q)|^2\right)\\&=\begin{cases}N^2\delta_{k,0}+(\kappa-1)(N-L), & \text{if}\ q=0,\\
	g^2_1(k,q)-L+q+N,    & \text{if}\  1\leq q< L,\\
    N,& \text{if}\  L\leq q\leq N-L,\\
    g^2_2(k,q)-L-q+2N,  &\text{if}\  N-L< q<N.
	\end{cases}	
\end{align} \it
When  $L> N/2$, the expected squared DP-AF is \rm
\begin{align}\label{block_dp_af_l_da} 
&\ \nonumber\mathbb{E}\left(|\mathcal{X}^{\mathrm{DP}}(k,q)|^2\right)\\&=\begin{cases}N^2\delta_{k,0}+(\kappa-1)(N-L), & \text{if}\ q=0,\\
	g^2_1(k,q)-L+q+N,    & \text{if}\  1\leq q\leq N-L,\\
    \multicolumn{2}{l}{\hspace*{-0.6em}\left[g_1(k,q)+(-1)^{k-uN}g_2(k,q)\right]^2+2N-2L,}\\&\text{if}\  N-L< q < L,\\
    g^2_2(k,q)-L-q+2N,  &\text{if}\  L\leq q<N,
	\end{cases}	
\end{align} \it
where $g_1(k,q)$ and $g_2(k,q)$ take the form of Dirichlet kernels \cite{dehkordi2023hierarchical} (a discrete-time analog of the sinc function):\rm
\begin{align}
    g_1(k,q):=\frac{\text{sin}\left[\pi\left(\frac{k}{N}-\frac{uq}{L}\right)(L-q)\right]}{\text{sin}\pi\left(\frac{k}{N}-\frac{uq}{L}\right)},
\end{align}
\begin{align}
    g_2(k,q):=\frac{\text{sin}\left[\pi\left(\frac{k}{N}-\frac{uq}{L}+\frac{uN}{L}\right)(L-N+q)\right]}{\text{sin}\pi\left(\frac{k}{N}-\frac{uq}{L}+\frac{uN}{L}\right)}.
\end{align} \it
If $k =\langle \frac{uqN}{L}\rangle_N$ or $k = \langle \frac{u(q-N)N}{L}\rangle_N$, the expected squared DP-AF in \eqref{block_dp_af_l_xiao} and \eqref{block_dp_af_l_da} approaches the corresponding limit.
\end{proposition}

\begin{proof}
    See Appendix \ref{proof_of_proposition_4}.
\end{proof}

\begin{figure}[t]
    \centering
    \begin{subfigure}[t]{.49\linewidth}
    \centering
    \includegraphics[width=1\linewidth]{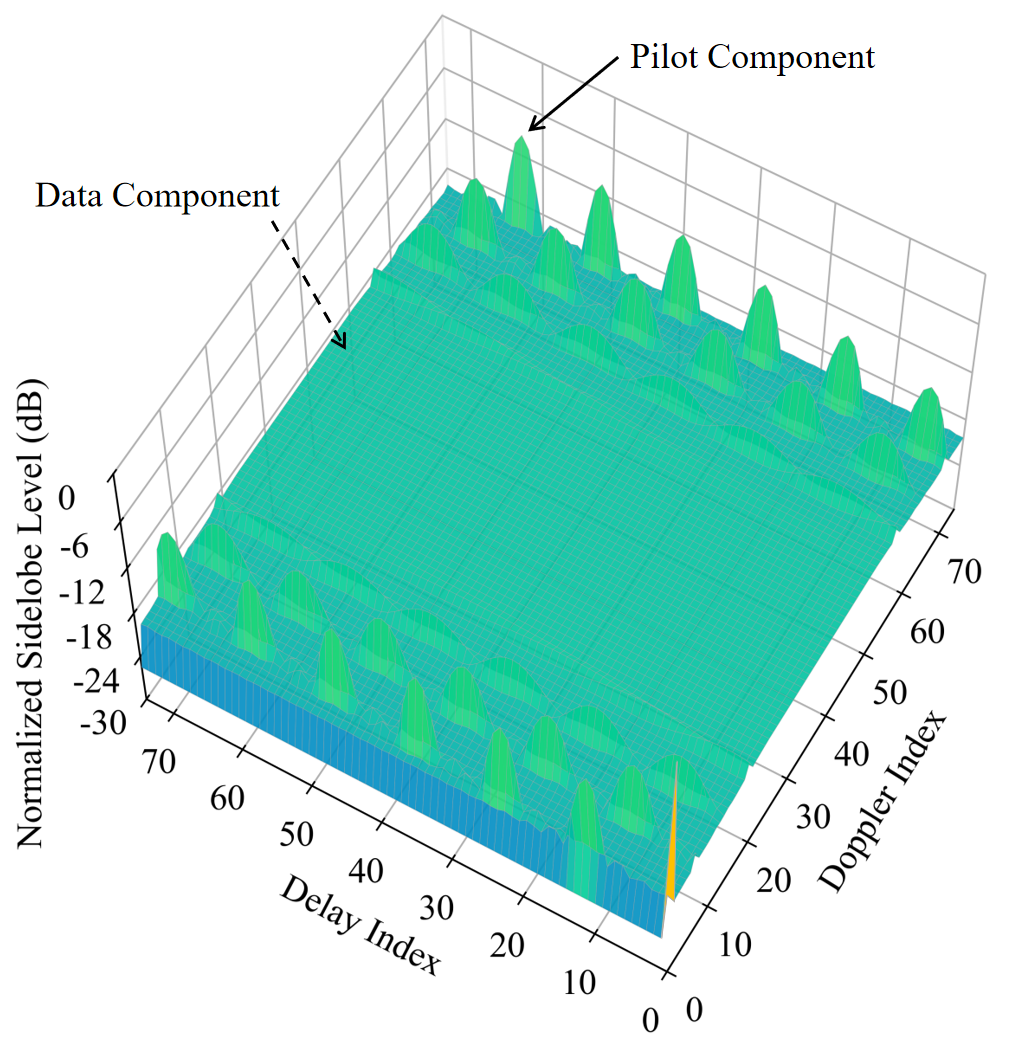}
    \caption{The case when $L=25$.}
    \label{2_1_block_l_xiao_theo_dpaf_3d_N_75_L_25_u_4}
    \end{subfigure}
    \centering
    \begin{subfigure}[t]{.49\linewidth}
    \centering
    \includegraphics[width=1\linewidth]{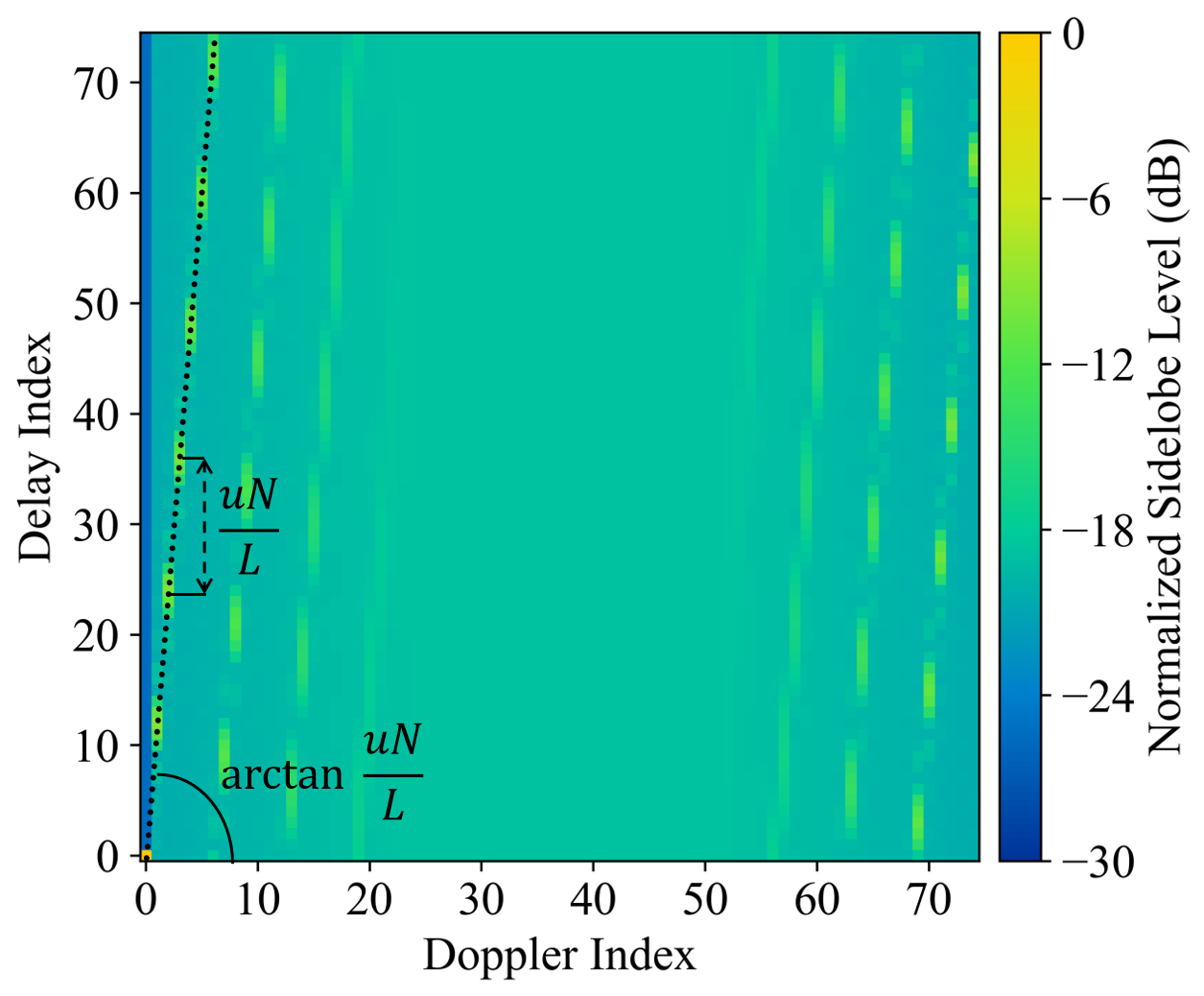}
    \caption{2D view of (a).}
    \label{fig:block_l_xiao_theo_dpaf_3d_and_2d}
    \end{subfigure}
    \centering
    \begin{subfigure}[t]{.49\linewidth}
    \centering
    \includegraphics[width=1\linewidth]{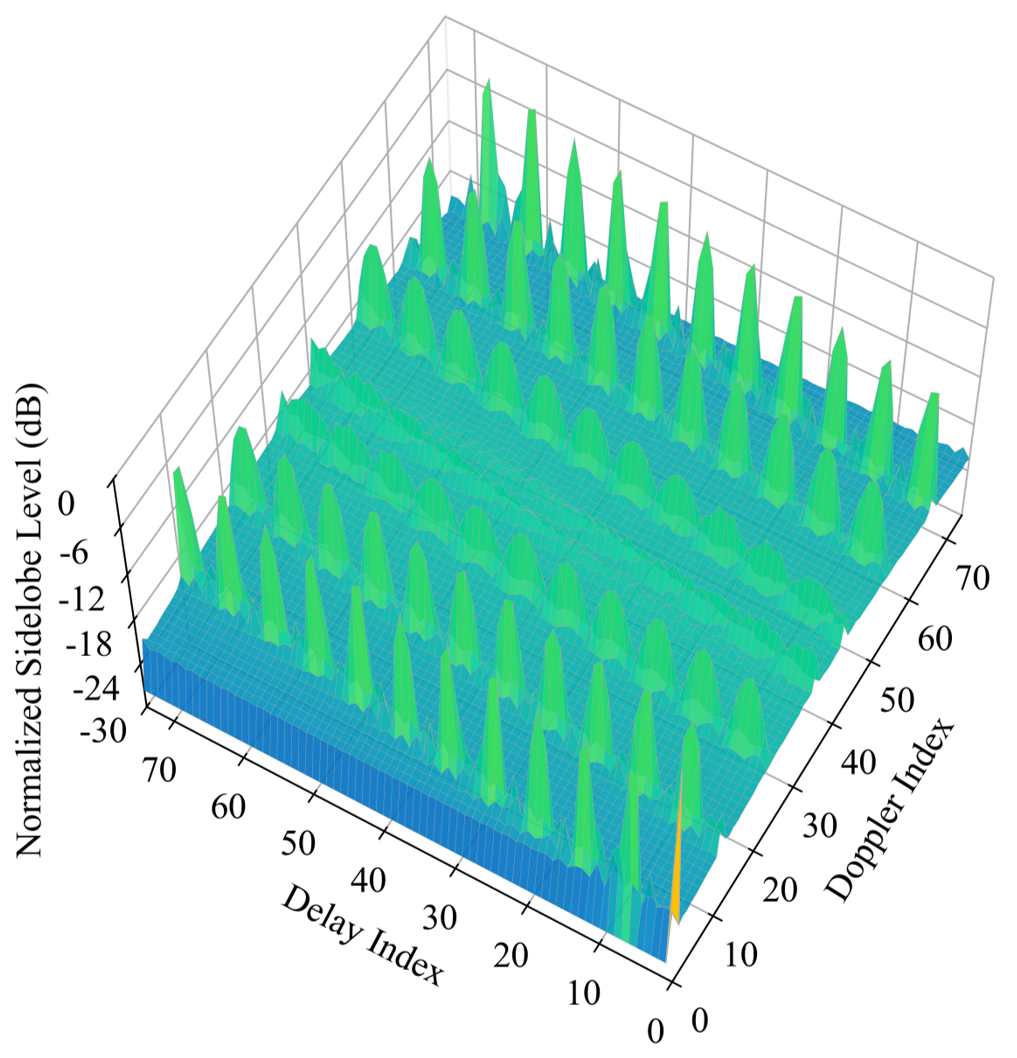}
    \caption{The case when $L=45$.}
    \end{subfigure}
    \centering
    \begin{subfigure}[t]{.49\linewidth}
    \centering
    \includegraphics[width=1\linewidth]{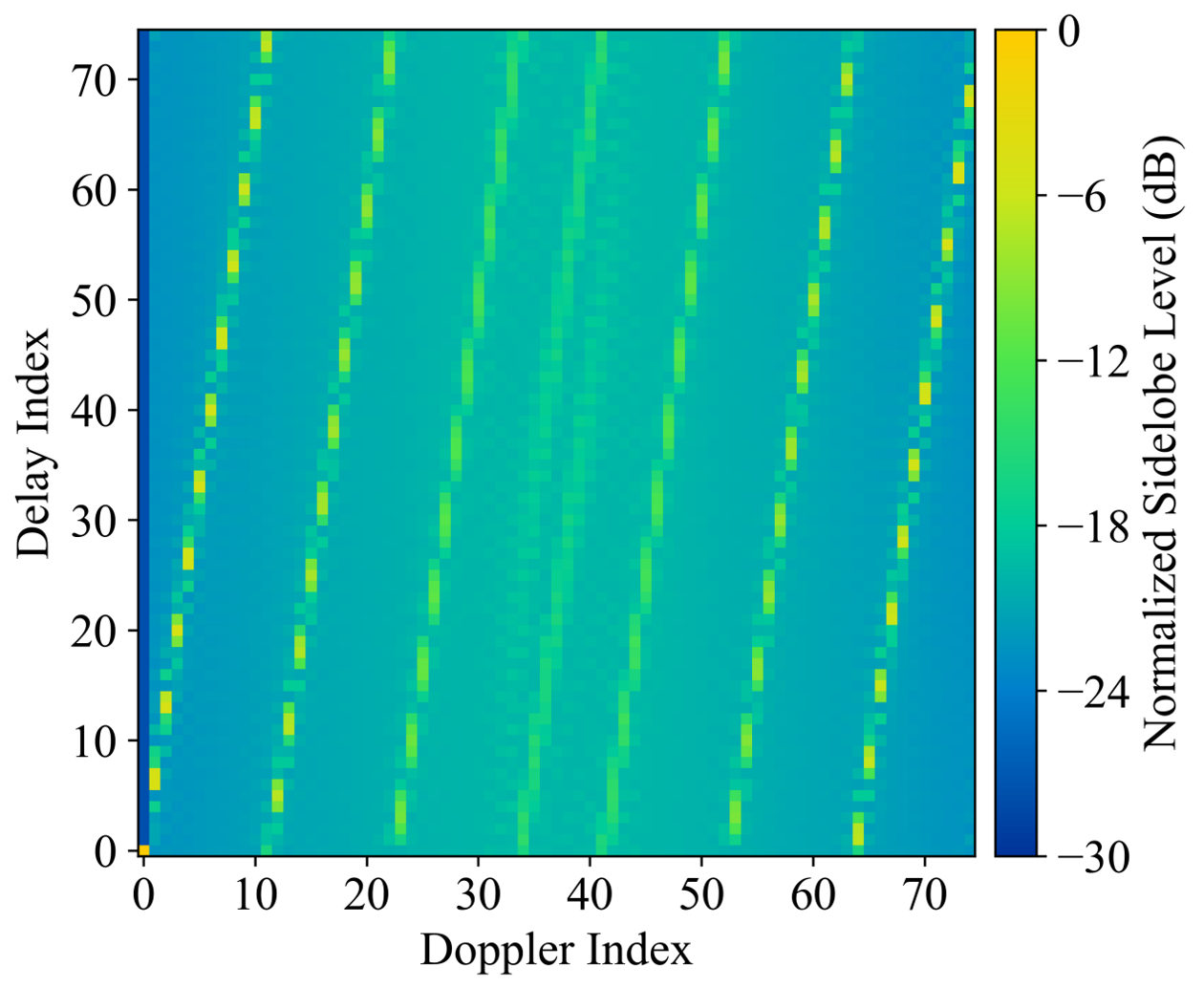}
    \caption{2D view of (c).}
    \label{fig:block_l_da_theo_dpaf_3d_and_2d}
    \end{subfigure}
    \caption{Theoretical expected squared DP-AF of the OFDM signal with contiguously placed ZC sequences, where $N=75$, $u=4$, under 16-QAM constellation.}
    \label{fig:all_block_theo_dpaf_3d_and_2d}
\end{figure}

\begin{figure*}[t]
    \centering
    \begin{subfigure}{.32\linewidth}
		\includegraphics[width=1\linewidth]{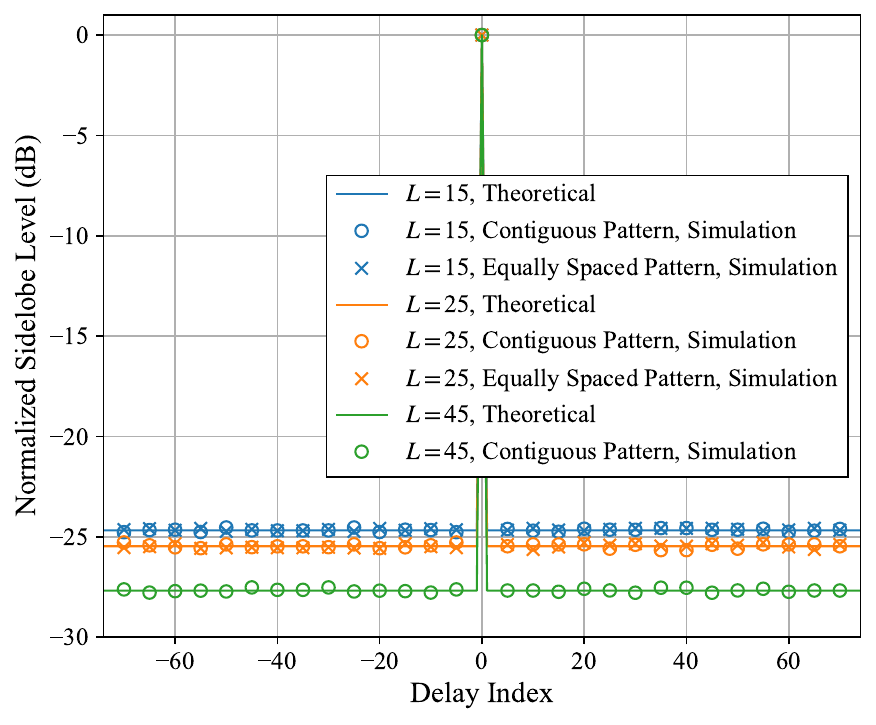}
		\caption{Zero-Doppler cuts.}
		\label{zero_doppler_slice}
	\end{subfigure}
	\begin{subfigure}{.32\linewidth}
		\includegraphics[width=1\linewidth]{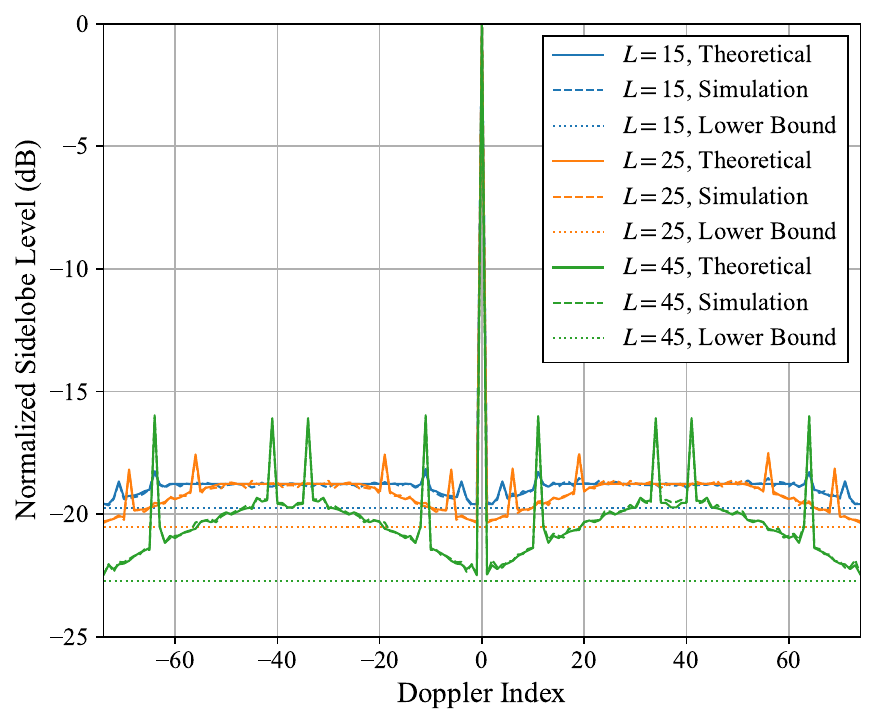}
		\caption{Zero-delay cuts (contiguous pattern).}
		\label{zero_delay_slice_cont_pattern}
	\end{subfigure}
    \begin{subfigure}{.32\linewidth}
		\includegraphics[width=1\linewidth]{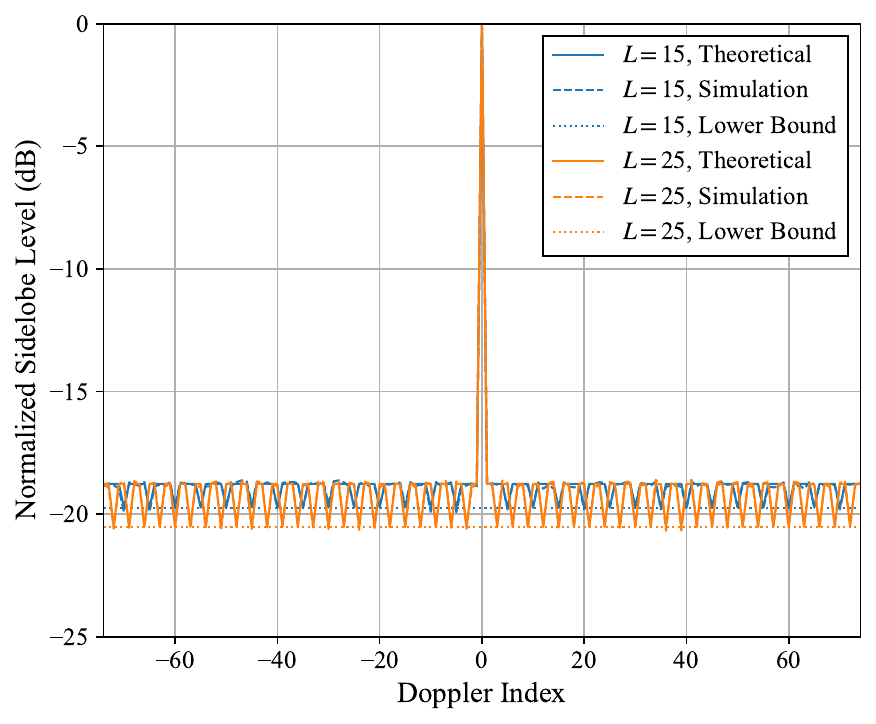}
		\caption{Zero-delay cuts (equispaced pattern).}
		\label{zero_delay_slice_equi_pattern}
	\end{subfigure}
	\caption{Zero-Doppler and zero-delay cuts of the theoretical and simulated expected squared DP-AF of the OFDM signal, with contiguously placed and equally spaced ZC sequences, where $N=75$, $u=4$, under 16-QAM constellation.} 
    \label{delay_and_doppler_slice}

    \hrulefill
\end{figure*}

When  $L\leq N/2$, the \textcolor{myLightRed}{non-zero} Doppler slices reveal two ridge-like regions, separated by a flat data-floor region. Specifically, for $1\leq q<L$, as shown in Fig. \ref{fig:block_l_xiao_theo_dpaf_3d_and_2d}, the DP-AF forms a ridge along the delay index $k$ that peaks near $k =\langle \frac{uqN}{L}\rangle_N$, tracing a linear delay-Doppler coupling with slope $uN/L$. \textcolor{myLightBlue}{The height} of this ridge, and hence its sharpness, is governed by the factor $(L-q)$ in $g_1$, reflecting the number of overlapping ZC subcarriers between the transmitted pilot and the Doppler-shifted received pilot. For $N-L< q<N$, an analogous ridge emerges from $g_2(k,q)$ due to the central symmetry stated in \eqref{central_symmetry}, now centered at $k = \langle \frac{u(q-N)N}{L}\rangle_N$, whose height is proportional to $(L-N+q)$. In the intermediate Doppler region $L\leq q\leq N-L$, the ZC pilot and its shifted replica share no overlapping subcarriers, \textcolor{myLightPurple}{thus} the AF \textcolor{myLightPurple}{reduces 
to} the constant data floor $N$, indicating the absence of pilot contribution.

When the pilot length exceeds half the bandwidth, i.e., $L>N/2$, as shown in Fig. \ref{fig:block_l_da_theo_dpaf_3d_and_2d}, the two ridge regions overlap in the interval $N-L<q<L$, giving rise to a coherent interference pattern between $g_1$ and $g_2$ that can either enhance or partially cancel the ridges depending on the parity of $(k-uN)$. 

\newpage

In the equally spaced case, the AF exhibits a periodic comb of sharp, isolated pulses. Here, the contiguous placement \textcolor{myLightRed}{reduces} the periodic structure into a continuous ridge that sweeps through the delay-Doppler plane along the coupling line. The design choice between the two patterns thus trades off 
between \textcolor{myLightRed}{concentrated sidelobe peaks and extended sidelobe ridges,}
with the baseline data-floor $N$ remaining identical in non-overlapping regions for both.

\textcolor{myLightRed}{\it Remark 1: \rm
The two structured pilot patterns analyzed above are designed primarily for communication tasks rather than for shaping the ambiguity response. As shown in Figs. \ref{fig:all_comb_theo_dpaf_3d_and_2d} and \ref{fig:all_block_theo_dpaf_3d_and_2d}, both patterns produce prominent high sidelobes in the expected squared DP-AF, arising from the coherent superposition of the pilot component with its Doppler-shifted replica. These high sidelobes cause strong sensing interference, so pilot patterns should be redesigned for ISAC. Owing to space limitations, we adopt a randomly placed pilot pattern as a simple low-complexity heuristic, and defer the systematic optimization of the pilot pattern for ISAC to future work. By spreading the pilot contributions irregularly across the delay-Doppler plane, the random pattern avoids the accumulation of sidelobe energy; its effectiveness in reducing the high sidelobes and the resulting sensing gain are verified in next section.}


\section{Simulation Results}\label{section_5}

This section provides numerical results that validate the theoretical findings in this work. All simulation outcomes are computed as averages over $3000$ independent runs. \newpage

\begin{figure}[t]
    \centering
    \includegraphics[width=0.8\linewidth]{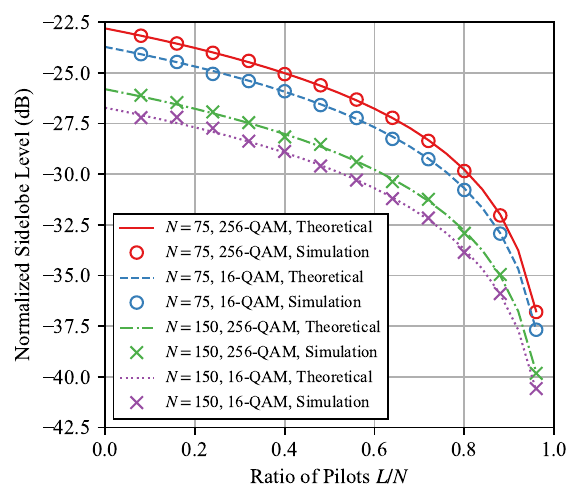}
    \caption{The normalized ESL on the zero-Doppler cuts of the expected squared DP-AF for the OFDM signal, with $N=75/150$ and varying $L$, under 16-QAM/256-QAM constellation.}
    \label{fig:DPAF_ACF_vs_pilot_num}
\end{figure}

Fig. \ref{delay_and_doppler_slice} compares the theoretical and simulated cuts of the expected squared DP-AF for the OFDM signal 
when $N=75$. The ZC sequences under two 
distinct 
pilot patterns are incorporated. The zero-Doppler cuts are first examined in Fig. \ref{zero_doppler_slice}. Since the zero-Doppler sidelobe behavior is independent of the pilot pattern, the curves for the contiguous and equally spaced configurations coincide. Increasing $L$ from $15$ to $45$ yields a 3 dB improvement in the peak-to-sidelobe level ratio. 
\textcolor{myLightPurple}{The agreement with simulations validates the analytical derivations.}

For the contiguous pattern, Fig. \ref{zero_delay_slice_cont_pattern} displays the zero-delay cuts for $L=15$, $25$ and $45$. All three curves exhibit dominant local peaks on each side of the mainlobe. This peak structure stems from the terms $g_1^2(0,q)$ and $g_2^2(0,q)$ and appears at Doppler indices where $uq/L$ or $u(N-q)/L$ approaches an integer. As $L$ grows, the peaks become higher while the valleys sink deeper. When $L$ is increased from $15$ to $45$, the minimum sidelobe level at zero-delay cut decreases by $2.86$ dB, at the cost of a $2.24$ dB increase in the maximum sidelobe level. The simulated curves closely match their theoretical counterparts. As predicted in \eqref{block_dp_af_l_xiao} and \eqref{block_dp_af_l_da}, in the cases $L=15$ and $25$ (where $L\leq N/2$), a flat data floor can be observed over the central interval, which vanishes for $L=45>N/2$.

For the equally spaced pattern in Fig. \ref{zero_delay_slice_equi_pattern}, the zero-delay cuts are displayed for $L=15$ and $L=25$, both satisfying $L\mid N$. As predicted by \eqref{zc_dp_af_zero_delay}, the profile exhibits a periodic comb of sharp, isolated dips precisely at Doppler indices $q = m(N/L)$. Each dip reaches the lower bound $(N-L)$, separated by intervals where the response stabilizes at the constant data floor $N$. This behavior is a direct consequence of the frequency-domain periodicity induced by the uniform pilot spacing. By tuning $L$, one can flexibly control both the density and the depth of the sensing notches along the Doppler axis while preserving an otherwise constant baseline.

\begin{figure}[t]
    \centering
    \begin{subfigure}{.49\linewidth}
		\includegraphics[width=1\linewidth]{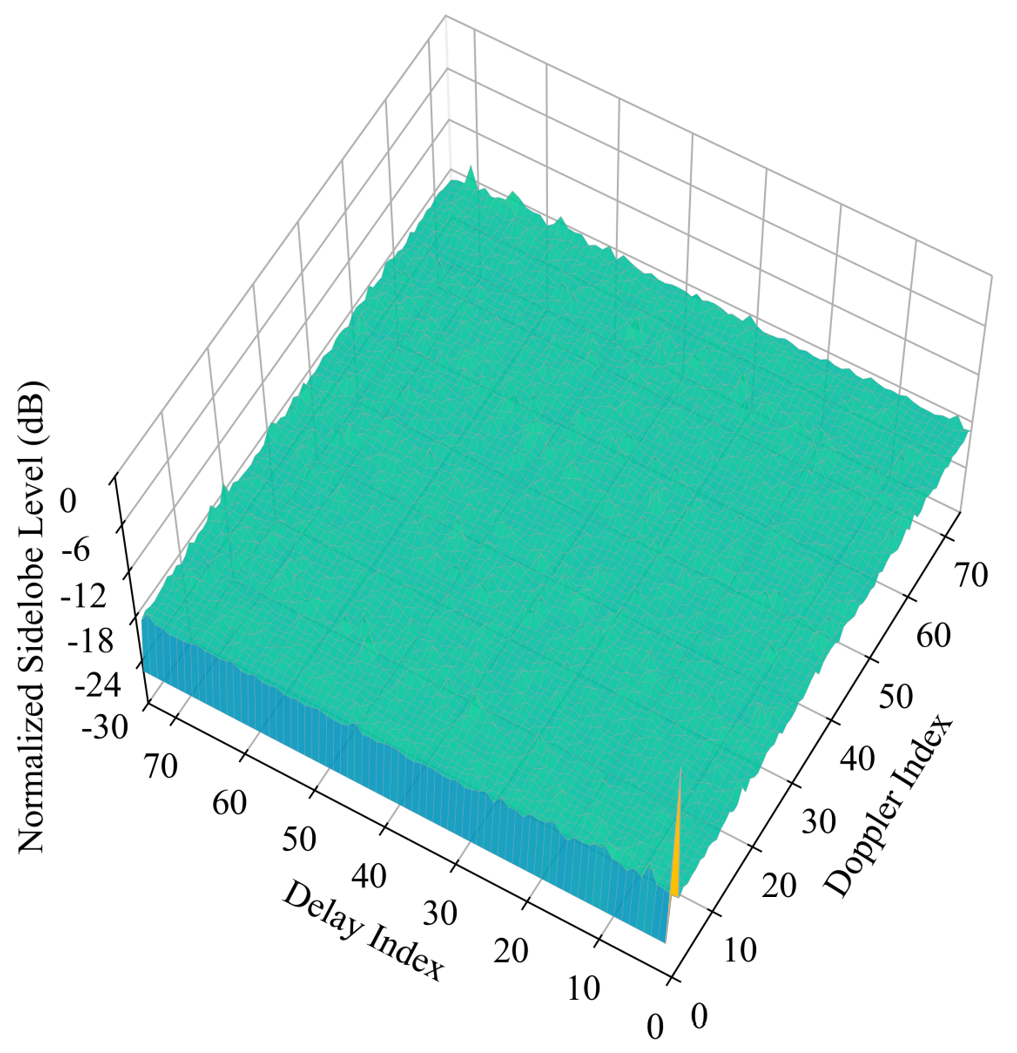}
		\caption{The case when $L=25$.}
	\end{subfigure}
    \centering
	\begin{subfigure}{.49\linewidth}
		\includegraphics[width=1\linewidth]{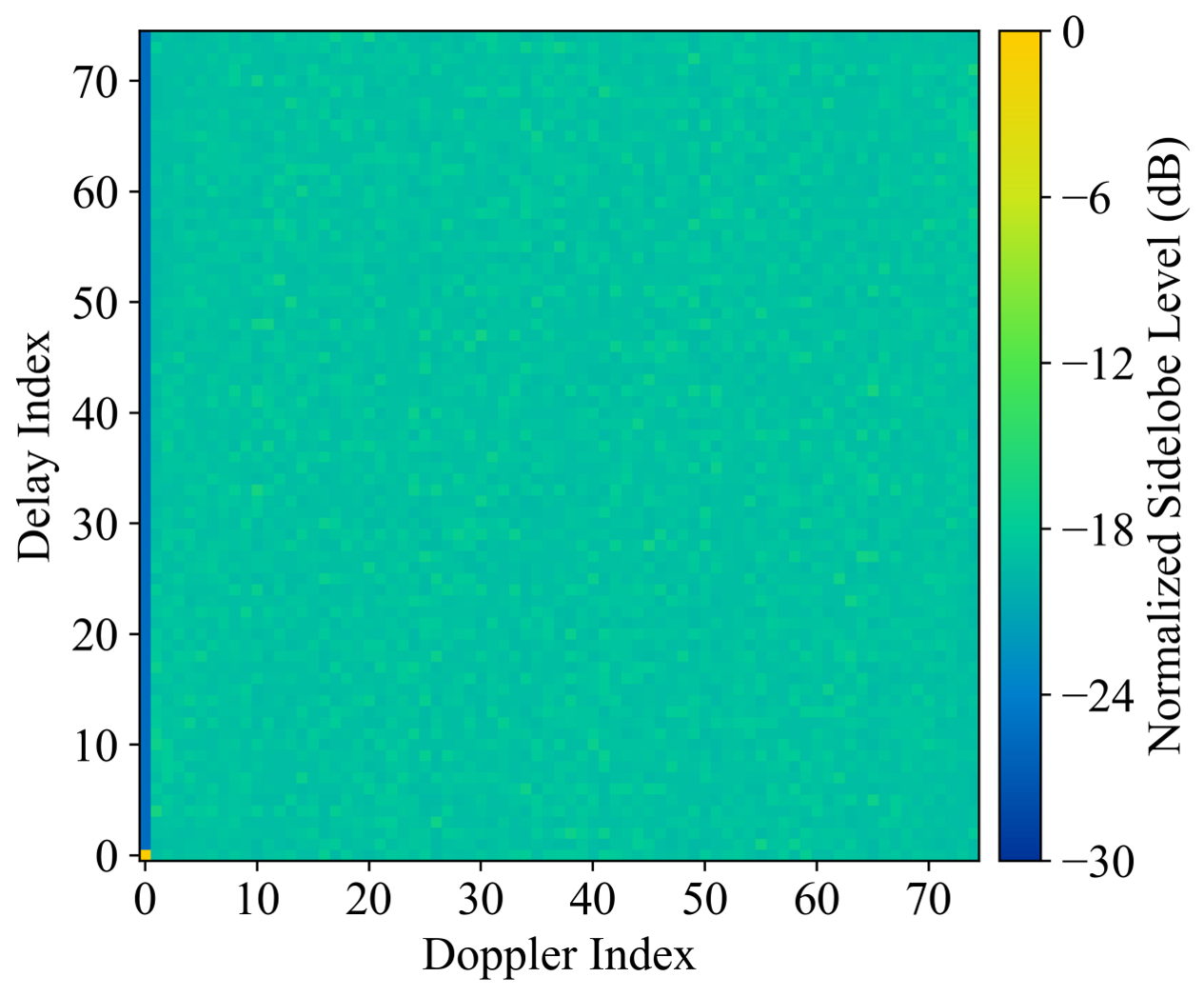}
		\caption{2D view of (a).}
	\end{subfigure}
	\caption{Theoretical expected squared DP-AF of the OFDM signal with a randomly placed ZC sequence, where $N=75$, $u=4$, under 16-QAM constellation.} 
    \label{fig:dpaf_random_pattern}
\end{figure}

\begin{figure}[t]
    \centering
    \begin{subfigure}{.99\linewidth}
		\includegraphics[width=1\linewidth]{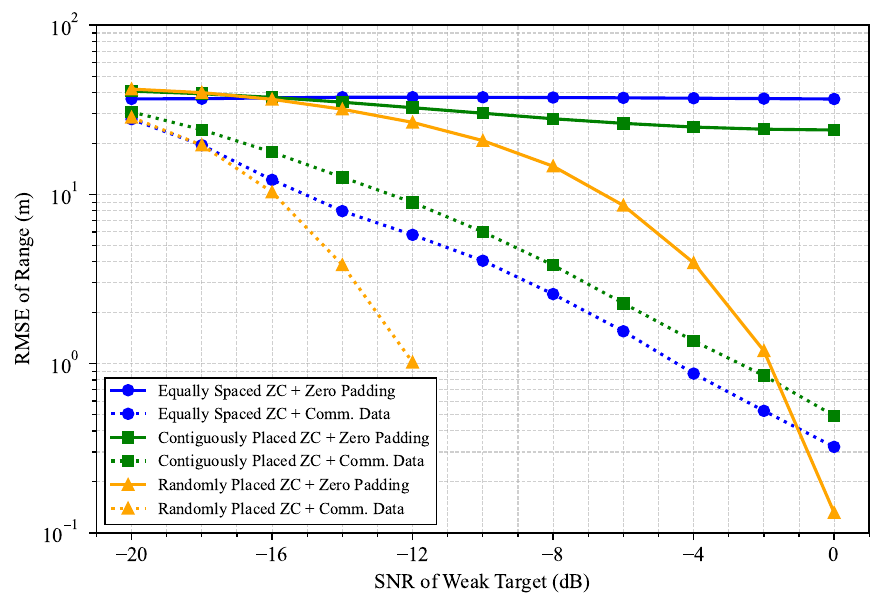}
		\caption{\textcolor{myLightRed}{RMSE of range.}}
	\end{subfigure}
    \centering
	\begin{subfigure}{.99\linewidth}
		\includegraphics[width=1\linewidth]{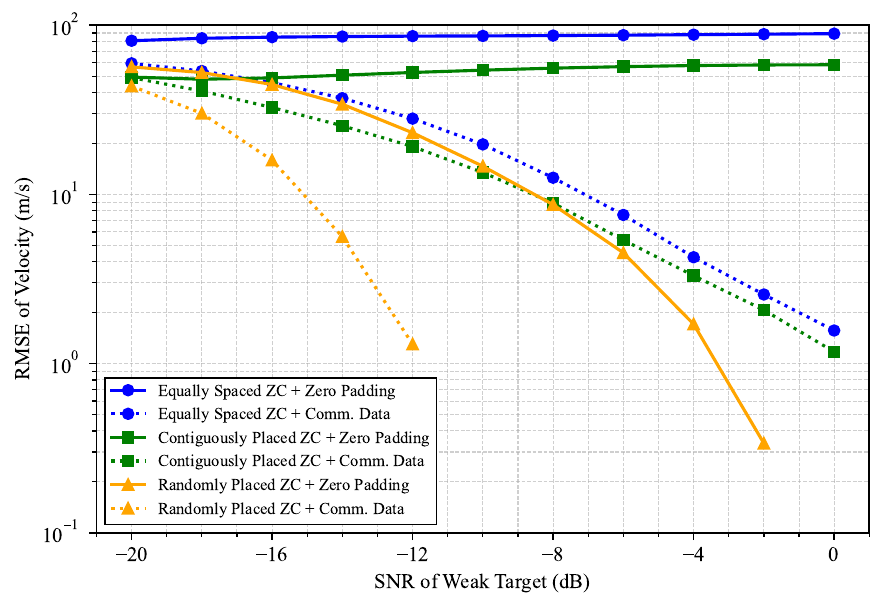}
		\caption{RMSE of velocity.}
	\end{subfigure}
	\caption{\textcolor{myLightBlue}{The RMSEs of the estimated range and velocity} for the OFDM signal with equally spaced, contiguously placed and randomly placed ZC sequences, where $N=417$, $L=139$, $u=1$, under 16-QAM constellation.} 
    \label{fig:RMSE}
\end{figure}

Next, we investigate the normalized ranging sidelobe levels on the zero-Doppler cuts of the expected squared DP-AF. In Fig. \ref{fig:DPAF_ACF_vs_pilot_num}, it can be seen that as the pilot ratio $L/N$ increases, the normalized ranging sidelobe levels decrease monotonically. By comparing the curves for $N=75$ and $N=150$ at the same $L/N$ ratio, it can be found that doubling $N$ provides approximately a $3$ dB reduction in the normalized ranging sidelobe level. Furthermore, 256-QAM exhibits nearly $1$ dB higher sidelobes than 16-QAM because of its larger kurtosis.

Fig. \ref{fig:dpaf_random_pattern} illustrates the expected squared DP-AF for a randomly drawn pilot pattern. In contrast to the DP-AF results in Figs. \ref{1_1_comb_theo_dpaf_3d_N_75_L_25_u_4} and \ref{2_1_block_l_xiao_theo_dpaf_3d_N_75_L_25_u_4} corresponding to the equally spaced and contiguous pilot patterns, the random placement is observed to reduce the high sidelobes. By spreading the pilot contributions irregularly across the delay-Doppler plane, the random placement avoids the accumulation of sidelobe energy, 
\textcolor{myLightPurple}{therefore} resulting in a flatter ambiguity profile.

For the 1D OFDM signal with its DP-AF already 
characterized, the estimation performance of range and velocity versus SNR is presented in Fig. \ref{fig:RMSE}. 
\textcolor{myLightRed}{The subcarrier spacing is $15$ kHz, and the carrier frequency is $77$ GHz.} 
The sensing region of interest corresponds to a range interval of $[0,119.9]$ m and a velocity interval of $[0,146.1]$ m/s. A strong target with amplitude $1$ is located at $(23.98\text{m},58.44\text{m/s})$, while a weak target with amplitude $0.5$ is located at $(71.94\text{m},29.22\text{m/s})$. SNR is defined as the amplitude of the weak target divided by the noise power. All root mean squared error (RMSE) results are evaluated with respect to the \textcolor{myLightPurple}{weak target.} 
The length of the ZC sequence is $139$, which is a typical setting of the short-format preamble of \textcolor{myLightPurple}{PRACH} in \textcolor{myLightPurple}{5G NR} \cite{tosun2026coordinated}. Under the same pilot pattern, the joint use of ZC sequences and random data for sensing yields a substantial reduction in RMSE compared to the case where only ZC sequences are employed, indicating that the utilization of random data can effectively improve sensing performance. Furthermore, the choice of the pilot pattern matters: the randomly placed pilot pattern outperforms both the equally spaced and the contiguous pilot patterns. The unsatisfactory performance of the equally spaced and contiguous patterns is attributed to the higher sidelobes associated with the corresponding expected squared DP-AF, which cause stronger interference from the strong target in the vicinity of the weak target.

\begin{figure}[t]
    \centering
    \begin{subfigure}{.49\linewidth}
		\includegraphics[width=1\linewidth]{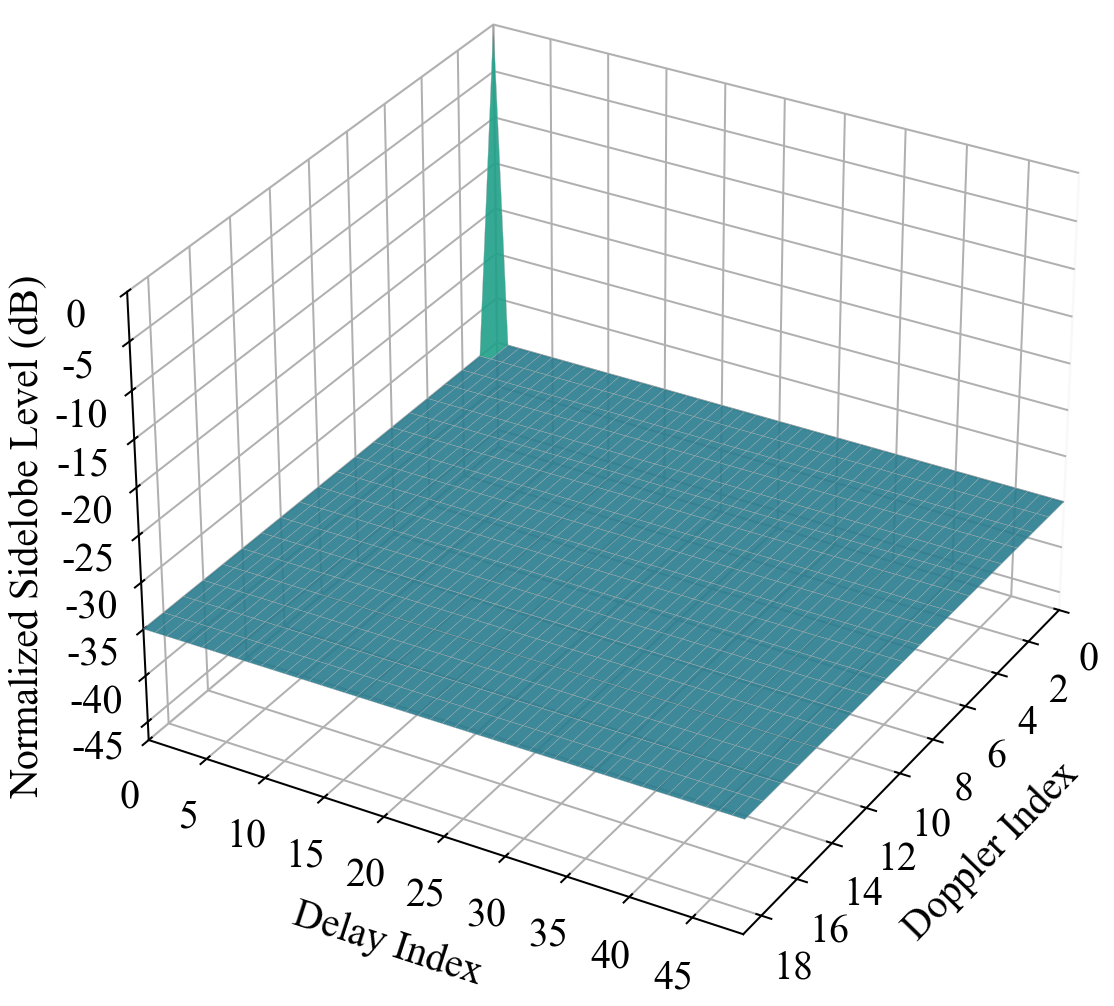}
		\caption{The case when $L=0$.}
		\label{fig:FSTAF_without_pilot_3D_theo}
	\end{subfigure}
    \centering
	\begin{subfigure}{.49\linewidth}
		\includegraphics[width=1\linewidth]{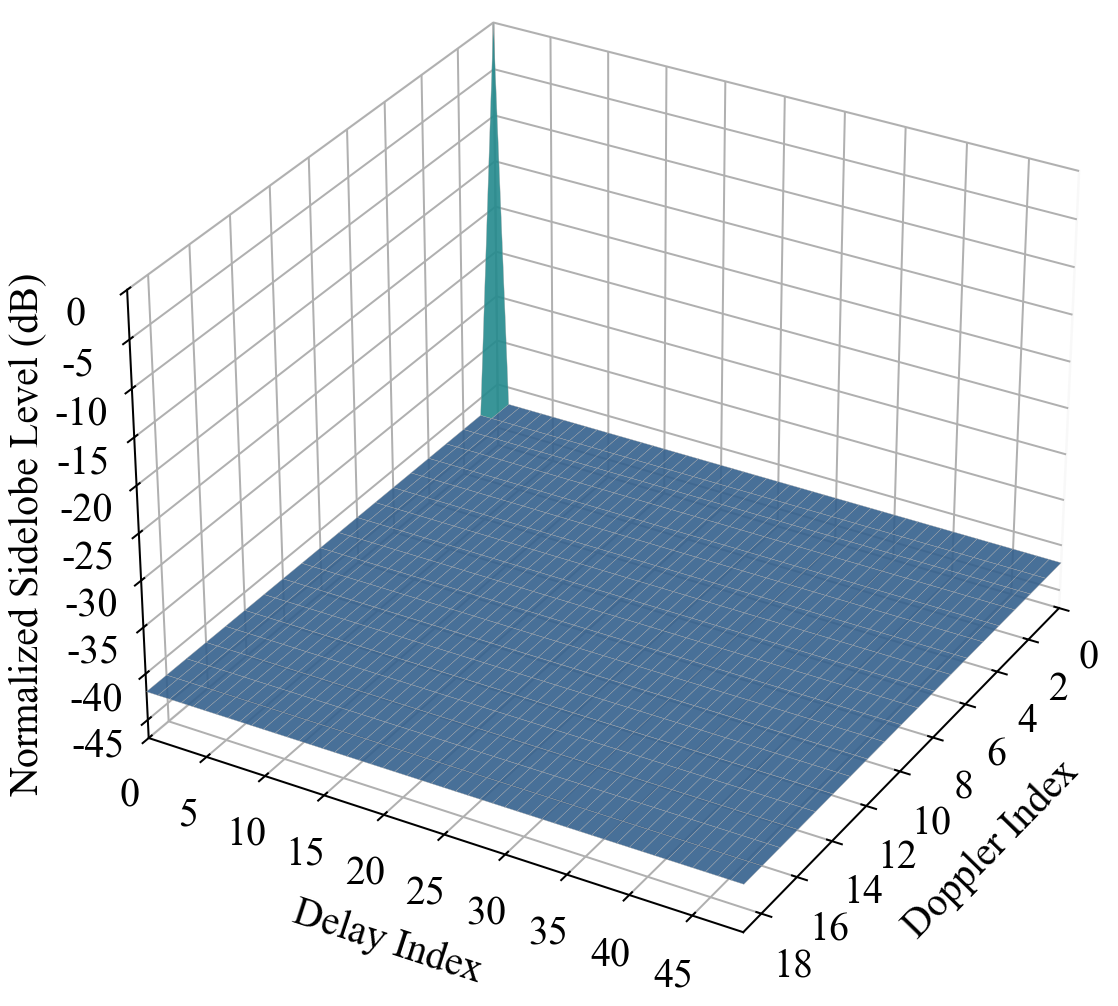}
		\caption{The case when $L=40M$.}
		\label{fig:FSTAF_with_pilot_3D_theo}
	\end{subfigure}
	\caption{Theoretical expected squared FST-AF of the OFDM signal, where $N=50$, $M=20$, under 16-QAM constellation.} 
    \label{fig:FSTAF_3D_theo}
\end{figure}

\begin{figure}[t]
    \centering
    \includegraphics[width=0.8\linewidth]{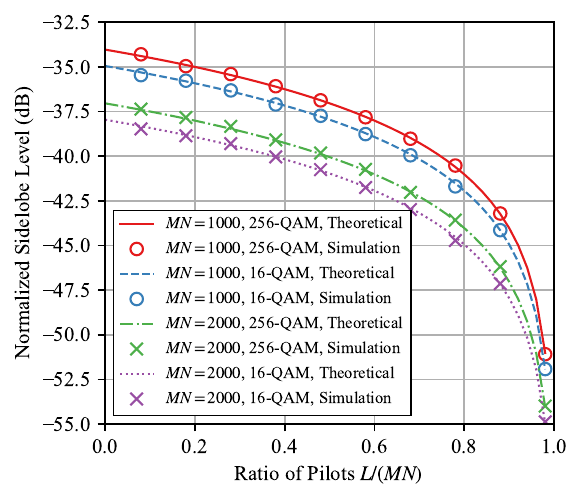}
    \caption{The normalized ESL of the expected squared FST-AF for the OFDM signal, with $N=50$, $M=20/40$ and varying $L$, under 16-QAM/256-QAM constellation.}
    \label{fig:FSTAF_sidelobe_vs_pilot_num}
\end{figure}

For the expected squared \textcolor{myLightPurple}{FST-AF}, it can be seen from Fig. \ref{fig:FSTAF_3D_theo} that the expected sidelobe levels are all equal. Similar to the DP-AF results in Fig. \ref{fig:DPAF_ACF_vs_pilot_num}, in Fig.~\ref{fig:FSTAF_sidelobe_vs_pilot_num}, it is observed that doubling $MN$ yields approximately a $3$ dB reduction in the normalized sidelobe level, and 256-QAM constellation produces higher sidelobes than 16-QAM due to its larger kurtosis.

\section{Conclusion}\label{section_6}

This paper has presented the analytical expressions for the expected squared DP-AF and FST-AF for pilot-embedded OFDM sensing. 
The expected sidelobes of FST-AF gives a uniform floor determined by the resource size, pilot count, and kurtosis.
The expected sidelobes of DP-AF depend on pilot symbols and patterns, \textcolor{myLightPurple}{and no pilot design can minimize all sidelobes in a 2D area simultaneously. Attaining the upper or lower bound at non-zero Doppler requires a periodic pilot pattern, with equally spaced chirp or ZC sequences maximizing the number of bound-achieving sidelobes. For ZC pilots, equally spaced placement yields a comb-like peak-and-notch profile, while contiguous placement produces delay-Doppler ridges described by squared Dirichlet kernels. Both regular patterns induce pronounced high sidelobes, and we introduce a randomly placed pattern to suppress the high sidelobes.}
\textcolor{myLightPurple}{Numerical simulations} 
confirmed the analytical results and demonstrated that jointly exploiting random data together with pilots yields substantial estimation gains, with the choice of pilot pattern markedly affecting the \textcolor{myLightPurple}{sensing} accuracy. Future work will incorporate communication-side metrics, 
such as the achievable rate and channel-estimation accuracy, 
into a joint ISAC optimization that co-designs the pilot and data \textcolor{myLightPurple}{resources.} 

{\appendices

\section{Proof of Theorem \ref{theorem_1}}\label{proof_of_theorem_1}

With the notation $\mathbf{B}=\mathbf{F}_N\mathbf{\Phi}_{N,q}\mathbf{R}_{N,k}\mathbf{F}_N^H$, \eqref{define_squared_DPAF} yields
\begin{align}\label{squared_DPAF_with_C}
    |\mathcal{X}^{\mathrm{DP}}(k,q)|^2  =|\mathbf{s}^H\mathbf{B}\mathbf{s}|^2=|(\mathbf{d}+\mathbf{p})^H\mathbf{B}(\mathbf{d}+\mathbf{p})|^2.
\end{align}
Let $m,n,l,r\in \mathbb{Z}_N$ and $w'_n=1-w_n$. From \eqref{eq:defdi}, it follows that $d_n$ takes a random communication symbol when $w_n=0$, and equals zero when $w_n=1$. Under Assumption \ref{assumption_1} we have
\begin{equation}
\mathbb{E}(d_m)=
\mathbb{E}(d_md_n)=0,\ 
\mathbb{E}(d_md_n^*)=w'_m\delta_{m,n},
\end{equation}
\begin{equation}
\mathbb{E}(d_m^*d_nd_l)=\mathbb{E}(d_m^*d_n^*d_l)=0,
\end{equation}
\begin{equation}\begin{aligned}
\mathbb{E}(d_m^*d_nd_ld_r^*)
&=w'_mw'_nw'_lw'_r\left(\delta_{m,n}\delta_{l,r}+\delta_{m,l}\delta_{n,r}\right)\\&
+(\kappa-2)\,w'_m\,\delta_{m,n,l,r}.
\end{aligned}\end{equation}
Hence the expectation of \eqref{squared_DPAF_with_C} can be decomposed as
\begin{align}\label{dpaf_expect_decompose}
\nonumber\mathbb{E}\left(|\mathcal{X}^{\mathrm{DP}}(k,q)|^2\right)
&=|\mathbf{p}^H\mathbf{B}\mathbf{p}|^2
+\mathbb{E}\left(|\mathbf{d}^H\mathbf{B}\mathbf{d}|^2\right)\\\nonumber&+\mathbb{E}\left(|\mathbf{p}^H\mathbf{B}\mathbf{d}+\mathbf{d}^H\mathbf{B}\mathbf{p}|^2\right)\\
&+
2\operatorname{Re}\!\left[(\mathbf{p}^H\mathbf{B}\mathbf{p})^*\,\mathbb{E}\left(\mathbf{d}^H\mathbf{B}\mathbf{d}\right)\right].
\end{align}
Due to $b_{m,n}=e^{-j2\pi nk/N}\delta_{n,\langle m-q\rangle_N}$, we have
\begin{equation}
    |\mathbf{p}^H\mathbf{B}\mathbf{p}|^2=\left|\sum_{n}p_np_{\langle n-q\rangle_N}^*\cdot e^{j2\pi nk/N}\right|^2.
\end{equation}
Let $\mathbf{W}'=\operatorname{Diag}(w'_0,\dots,w'_{N-1})$, then the remaining items in \eqref{dpaf_expect_decompose} can be derived as 
\begin{equation}
\mathbb{E}\left(\mathbf{d}^H\mathbf{B}\mathbf{d}\right)
=\operatorname{tr}(\mathbf{B}\mathbf{W}')=\sum_n w'_ne^{-j2\pi nk/N}\,\delta_{q,0},
\end{equation}
\begin{equation}\begin{aligned}
\mathbb{E}\left(|\mathbf{p}^H\mathbf{B}\mathbf{d}+\mathbf{d}^H\mathbf{B}\mathbf{p}|^2\right)
&=\sum_n w'_n\left(|(\mathbf{B}\mathbf{p})_n|^2+|(\mathbf{B}^H\mathbf{p})_n|^2\right)\\&=\sum_n w'_n\left(w_{\langle n-q\rangle_N}+w_{\langle n+q\rangle_N}\right),
\end{aligned}
\end{equation}
\begin{align}
\nonumber\mathbb{E}\left(|\mathbf{d}^H\mathbf{B}\mathbf{d}|^2\right)
&=\left|\operatorname{tr}(\mathbf{B}\mathbf{W}')\right|^2
+\operatorname{tr}\left(\mathbf{B}\mathbf{W}'\mathbf{B}^H\mathbf{W}'\right)\\\nonumber
&+(\kappa-2)\sum_n |b_{nn}|^2 w'_n
\\\nonumber&=\left|\sum_n w'_ne^{-j2\pi nk/N}\,\delta_{q,0}\right|^2+\sum_n w'_nw'_{\langle n-q\rangle_N}
\\&+(\kappa-2)\sum_n  w'_n\,\delta_{q,0}.
\end{align}
Thus, the expected squared DP-AF when $q\neq0$ is derived as
\begin{align}
\nonumber\mathbb{E}\left(|\mathcal{X}^{\mathrm{DP}}(k,q)|^2\right)
&=|\mathbf{p}^H\mathbf{B}\mathbf{p}|^2+\sum_n w'_nw'_{\langle n-q\rangle_N}
\\&\nonumber+\sum_n w'_n\left(w_{\langle n-q\rangle_N}+w_{\langle n+q\rangle_N}\right)
\\&=|\mathbf{p}^H\mathbf{B}\mathbf{p}|^2+N-\sum_n w_nw_{\langle n-q\rangle_N}.
\end{align}
And when $q=0$, the expected squared DP-AF is derived as
\begin{align}
\nonumber\mathbb{E}\left(|\mathcal{X}^{\mathrm{DP}}(k,0)|^2\right)
&=\left|\sum_{n}w_ne^{j2\pi nk/N}\right|^2+\left|\sum_{n}w'_ne^{-j2\pi nk/N}\right|^2
\\\nonumber&+2\operatorname{Re}\!\left[\sum_n w_ne^{j2\pi nk/N}\,\sum_n w'_ne^{-j2\pi nk/N}\right]
\\&\nonumber+(\kappa-1)(N-L)
\\&=N^2\delta_{k,0}+ (\kappa-1)(N-L),
\end{align}
which matches the expression in Theorem \ref{theorem_1}.

\section{Proof of Proposition \ref{proposition_2}}\label{proof_of_proposition_2} 
Define $\mathcal{W}=\{n\mid w_n=1\}$, and $\mathcal{W}+q=\{\langle n+q\rangle_N\mid n\in \mathcal{W}\}.$ Then $h(q)=\mathrm{Card}[\mathcal{W}\cap(\mathcal{W}+q)]=\mathrm{Card}[\mathcal{W}\cap(\mathcal{W}-q)].$ Note that $h(q)=L$ is equivalent to $\mathcal{W}+q=\mathcal{W}$ and $\mathcal{W}-q=\mathcal{W}$, thus $\mathcal{D}_h\cup\{0\}$ is closed under addition and subtraction. 


When $\mathcal{D}_h\neq\varnothing$, assume $q_{min}$ is the minimal positive element in $\mathcal{D}_h\cup\{0\}$. For any $e\in \mathcal{D}_h$, we have $e=r_1q_{min}+r_2$, where $r_1,r_2\in\mathbb{Z}, 0\leq r_2<q_{min}$. If $r_2\neq0$, by the closure of $\mathcal{D}_h\cup\{0\}$, it follows that $r_2\in\mathcal{D}_h$, contradicting the minimality of $q_{min}$. Thus all elements in $\mathcal{D}_h$ are multiples of $q_{min}$, i.e., 
\begin{align}
    \mathcal{D}_h\cup\{0\}=\{iq_{min}|i=0,...,\mathrm{Card}\left(\mathcal{D}_h\right)\}.
\end{align}

Furthermore, it is constrained that $\mathrm{Card}\left(\mathcal{D}_h\right)q_{min}<N$ and $[\mathrm{Card}\left(\mathcal{D}_h\right)+1]q_{min}\geq N$. Due to the closure of $\mathcal{D}_h\cup\{0\}$, $\langle [\mathrm{Card}\left(\mathcal{D}_h\right)+1]q_{min}\rangle_N\in\mathcal{D}_h\cup\{0\}$, which leads to \textcolor{myLightPurple}{$[\mathrm{Card}\left(\mathcal{D}_h\right)+1]q_{min}=N$.}

Due to $q_{min}\in\mathcal{D}_h$, for any $m=1,...,[\mathrm{Card}\left(\mathcal{D}_h\right)+1]$, $\mathcal{W}+mq_{min}=\mathcal{W}$.
Thus for any $a\in\mathcal{W}$, $\langle a+mq_{min}\rangle_N\in\mathcal{W}$, constituting $\mathrm{Card}\left(\mathcal{D}_h\right)+1$ elements in $\mathcal{W}$, which leads to \textcolor{myLightPurple}{$[\mathrm{Card}\left(\mathcal{D}_h\right)+1]\mid L$}.

Hence, when $\mathcal{D}_h\neq\varnothing$, $\mathrm{Card}\left(\mathcal{D}_h\right)+1>1$ is a common divisor of $N$ and $L$, and $\mathbf{w}$ has period $q_{min}=N/[\mathrm{Card}\left(\mathcal{D}_h\right)+1]$. Conversely, if $\mathrm{gcd}(N,L)>1$ and $\lambda_{N,L}>1$ is a common divisor of $N$ and $L$, we can construct a pattern $\mathbf{w}$ with period $T_{\mathbf{w}}=N/\lambda_{N,L}$, which satisfies $\mathcal{D}_h=\{iT_\mathbf{w}|i=1,...,\lambda_{N,L}-1\}$:
\begin{align}
    w_i=\begin{cases}
        1, & \text{if $\langle i\rangle_{T_\mathbf{w}}=0,...,L/\lambda_{N,L}-1$},
        \\0, & \text{else}.
    \end{cases}
\end{align}

\section{Proof of Proposition \ref{proposition_3}}\label{proof_of_proposition_3}
For the $N$ sidelobes corresponding to any $q\in \mathcal{D}_h$, owing to the invariance in \eqref{sum_of_sidelobe}, it can be verified that if the number of sidelobes that attain the lower bound is $(N-N/L)$, then the remaining $N/L$ sidelobes attain the upper bound. 
\textcolor{myLightPurple}{In this case}, the number of sidelobes that attain the lower or upper bounds cannot be increased further. Additionally, if $\mathrm{Card}\left(\mathcal{D}_h\right)=L-1$, 
then $L\mid N$, $\mathcal{D}_h=\{m(N/L)|m=1,...,L-1\}$ and that the pilot pattern is equally spaced with a period of $N/L$.

Next, we derive the pilot symbols. Without loss of generality, 
\textcolor{myLightPurple}{let the pilot subcarriers be} 
$l\frac{N}{L}$, $l=0,...,L-1$. Define $c_l=p_{l\frac{N}{L}}$ and $v_l^{(m)}=c_lc_{\langle l-m\rangle_L}^*$, which leads to
\begin{align}
    g(k,q)=\begin{cases}
        \left|\sum_{l=0}^{L-1}v_l^{(m)}e^{j2\pi l\frac{k}{L}}\right|^2& \text{if $q=m\frac{N}{L}$}\in \mathcal{D}_h,\\
        0,&\text{else}.
    \end{cases}
\end{align}
Accordingly, define the Fourier transform pair as follows:
\begin{equation}
    \begin{aligned}\label{transform_pair}
    V_k^{(m)}=\sum_{l=0}^{L-1}v_l^{(m)}e^{j2\pi l\frac{k}{L}},\ v_l^{(m)}=\frac{1}{L}\sum_{k=0}^{L-1}V_k^{(m)}e^{-j2\pi l\frac{k}{L}}.
\end{aligned}
\end{equation}

For $q=\frac{N}{L}\in \mathcal{D}_h$, $g(k,q)$ attains $L^2$ for $N/L$ values of $k$, and $0$ for the rest. Hence, within one period (i.e., for $k\in \mathbb{Z
}_L$), only one $k$ (denoted as $u$) gives $L^2$, while the others give $0$. Thus we have 
\textcolor{myLightPurple}{$V_k^{(1)}=Le^{j\theta_0}\delta_{k,u},\ \forall k\in \mathbb{Z}_L$,} 
where $\theta_0$ is the phase of $V_u^{(1)}$. By \textcolor{myLightPurple}{leveraging} 
\eqref{transform_pair}, we have 
\textcolor{myLightPurple}{$v_l^{(1)}=e^{j\theta_0}e^{-j2\pi l\frac{u}{L}},\ \forall l=0,...,L-1
$}, 
which implies \textcolor{myLightPurple}{that} 
\begin{align}\label{c_ditui}
c_l=v_l^{(1)}c_{\langle l-1\rangle_L}=e^{j\theta_0}e^{-j2\pi l\frac{u}{L}}c_{\langle l-1\rangle_L}.
\end{align}
Assume $c_0=e^{j\phi_0}$. It follows recursively from \eqref{c_ditui} that
\begin{align}\label{c_l}
    c_l=e^{j\phi_0}e^{-j\pi \frac{u}{L}l(l+1)}e^{jl\theta_0}, \ \forall l=0,...,L-1.
\end{align}
Due to $c_0=e^{j\theta_0}c_{L-1}$, we have $e^{jL\theta_0}=e^{j\pi u(L-1)}$, hence
\begin{align}\label{theta_0}
    \exists\  T\in \mathbb{Z},\ L\theta_0=\pi u(L-1)+2T\pi.
\end{align}
Substituting \eqref{theta_0} into \eqref{c_l} yields 
\begin{align}
    c_l=e^{j\phi_0}e^{-j\pi \frac{u}{L}l(l+1)}e^{jl\frac{\pi}{L}[u(L-1)+2T]}.
\end{align}

When $L$ is odd, let $A=e^{j\phi_0}$, $\nu=T+u(L-1)/2\in \mathbb{Z}$; then the expression of pilot symbols is
\begin{align}\label{chirp_odd_fulu}
    p_{l\frac{N}{L}}=c_l=A\cdot e^{\left[-j\frac{\pi u}{L}l(l+1)\right]}\cdot e^{\left(j\frac{2\pi\nu}{L}l\right)}.
\end{align}

When $L$ is even, let $A=e^{j\phi_0}$, $\nu=T+u(L-2)/2\in \mathbb{Z}$; then the expression of pilot symbols is
\begin{align}\label{chirp_even_fulu}
    p_{l\frac{N}{L}}=c_l=A\cdot e^{\left(-j\frac{\pi u}{L}l^2\right)}\cdot e^{\left(j\frac{2\pi\nu}{L}l\right)}.
\end{align}

Next, we prove that for any $A\in\mathbb{C}$, $|A|=1$, and $u,\nu \in \mathbb{Z}$, the numbers of sidelobes that attain the upper and lower bounds are maximized. \textcolor{myLightPurple}{Substituting \eqref{chirp_odd_fulu} into $V_k^{(m)}=\sum_{l=0}^{L-1}p_{l\frac{N}{L}}p_{\langle l-m\rangle_L\frac{N}{L}}^*e^{j2\pi l\frac{k}{L}}$ yields}
\textcolor{myLightPurple}{\begin{align}
& \nonumber
V_k^{(m)}=\sum_{l=0}^{m-1}e^{-j\frac{\pi ul(l+1)}{L}}e^{j\frac{\pi u(l-m+L)(l-m+L+1)}{L}}e^{j\frac{2\pi\nu}{L}m}e^{j2\pi l\frac{k}{L}}\\& +\sum_{l=m}^{L-1}e^{-j\frac{\pi ul(l+1)}{L}}e^{j\frac{\pi u(l-m)(l-m+1)}{L}}e^{j\frac{2\pi\nu}{L}m}e^{j2\pi l\frac{k}{L}}
\end{align}}
\newpage
\textcolor{myLightPurple}{which leads to}
\begin{align}
&\nonumber\  V_k^{(m)}
=\begin{cases}
    L e^{j\pi\frac{u}{L}m(m-1)}e^{j\frac{2\pi\nu}{L}m}, & \text{if} \ L\mid(k-um)\\0, &\text{else}
\end{cases}.
\end{align}

Similarly, substituting \eqref{chirp_even_fulu} into $V_k^{(m)}$ yields
\begin{align}
&\ V_k^{(m)}
=\begin{cases}
    Le^{j\pi\frac{u}{L}m^2}e^{j\frac{2\pi\nu}{L}m}, & \text{if} \ L\mid(k-um)\\0, &\text{else}
\end{cases}.
\end{align}
Consequently we have
\begin{align}\label{g_in_comb}
&\ g(k,q) = \begin{cases}
    L^2, & \text{if}\ q=m\frac{N}{L}\  \text{and} \ L\mid(k-um)\\0, &\text{else}
\end{cases}.
\end{align}
Furthermore, the expression of $h(q)$ can be given as
\begin{align}\label{h_in_comb}
h(q)=\begin{cases} L, \quad & \text{if}\ q=m\frac{N}{L}(m=1,...,L-1)\\0,&\text{else},\end{cases}.\end{align}

By combining the results in \eqref{g_in_comb} and \eqref{h_in_comb} \textcolor{myLightPurple}{via} \eqref{f_g_h_relation}, we arrive \textcolor{myLightPurple}{at \eqref{chirp_dp_af}}. One can easily check that the equality case holds here.

\section{Proof of Theorem \ref{theorem_2}}\label{proof_of_theorem_2}
With the notation $\mathbf{A}=\mathbf{F}_N^H$,  $\mathbf{A}'=\mathbf{F}_M$, \eqref{define_FSTAF} gives:
\begin{align}\mathcal{X}^{\mathrm{FST}}(k,q)&=\sqrt{MN}\sum_{n=0}^{N-1}\sum_{m=0}^{M-1}a_{k,n}a'_{m,q}\left|s_{n,m}\right|^2,
\end{align}
then we have
\begin{align}
&\nonumber\ |\mathcal{X}^{\mathrm{FST}}(k,q)|^2 = \sum_{n=0}^{N-1}\sum_{m=0}^{M-1} \left|s_{n,m}\right|^4\\&+MN\sum_{(n,m)\neq(l,r)}a_{k,n}a_{k,l}^*a'_{m,q}a_{r,q}'^{*}\left|s_{n,m}\right|^2\left|s_{l,r}\right|^2.
\end{align}
The corresponding expectation is
\begin{align}
&\nonumber\ \mathbb{E}\left(|\mathcal{X}^{\mathrm{FST}}(k,q)|^2\right)\nonumber\\&=\kappa(MN-L)+L\nonumber+MN\sum_{(n,m)\neq(l,r)}a_{k,n}a_{k,l}^*a'_{m,q}a_{r,q}'^*
\\&\nonumber=
(\kappa-1)(MN-L)\\&\nonumber+\sum_{n,m,l,r}\text{cos}\left[\frac{2\pi k}{N}(n-l)-\frac{2\pi q}{M}(m-r)\right],
\\&=(\kappa-1)(MN-L)+M^2N^2\delta_{k,0}\delta_{q,0},
\end{align}
which matches the expression in Theorem \ref{theorem_2}.

\section{Proof of Proposition \ref{proposition_4}}\label{proof_of_proposition_4} 
First we consider the situation when $L\leq{N}/2$. Without loss of generality, 
\textcolor{myLightPurple}{let the pilot subcarriers be} 
$0,...,L-1$. When $q>N-L$, 
\textcolor{myLightPurple}{$g(k, q)$ is derived as:}
\begin{equation}\label{block_g_l_xiao}
    \begin{aligned} 
&\ 
\textcolor{myLightPurple}{\left|\sum_{n=0}^{q-(N-L+1)}e^{-j\frac{\pi un(n+1)}{L}}e^{j\frac{\pi u(n-q+N)(n-q+N+1)}{L}}\cdot e^{j2\pi nk/N}\right|^2}
\\&=\begin{cases} (L-N+q)^2, & \text{if}\ k=\langle \frac{u(q-N)N}{L}\rangle_N,\\
\frac{\text{sin}^2\left[\pi\left(\frac{k}{N}-\frac{uq}{L}+\frac{uN}{L}\right)(L-N+q)\right]}{\text{sin}^2\pi\left(\frac{k}{N}-\frac{uq}{L}+\frac{uN}{L}\right)}&\text{else}.\end{cases}
\end{aligned}
\end{equation}
The case when $q<L$ can be obtained through the central symmetry in \eqref{central_symmetry}. Furthermore, $h(q)$ can be given as
\textcolor{myLightPurple}{\begin{align}
    h(q)=\max\{L-q,\mathrm{~}0,\mathrm{~}q-(N-L)\},
\end{align}}
\textcolor{myLightPurple}{and combining the above results via \eqref{f_g_h_relation} gives \eqref{block_dp_af_l_xiao}.}

Next we consider the situation when $L>{N}/2$. The \textcolor{myLightPurple}{cases} when $1\leq q \leq N-L$ and $L\leq q <N$ 
\textcolor{myLightPurple}{follow as in} 
\eqref{block_g_l_xiao}. When $N-L<q<L$, $g(k, q)$ is derived as:
\begin{equation}
    \begin{aligned}
& \left|\sum_{n=0}^{q-1}p_np^*_{n-q+N}\cdot e^{j2\pi nk/N}+\sum_{n=q}^{L-1}p_np^*_{n-q}\cdot e^{j2\pi nk/N}\right|^2=\\&
\left|e^{\frac{j\pi u(N-q)(N-q+1)}{L}}\sum_{n=0}^{n_q}e^{j2\pi n\beta}+e^{\frac{j\pi uq(q-1)}{L}}\sum_{n=q}^{L-1}e^{j2\pi n\alpha}\right|^2.
\end{aligned}
\end{equation}
where $n_q=L-1+q-N$, $\alpha = \frac{k}{N}-\frac{qu}{L}$ and $\beta = \alpha+\frac{uN}{L}$. When $e^{j2\pi n\alpha} \neq 1$ and $e^{j2\pi n\beta} \neq 1$, 
$g(k,q)$ can be further derived as
\begin{equation}\begin{aligned}
\textcolor{myLightPurple}{\left|\frac{\text{sin}\left[\pi\left(L-q\right)\alpha\right]}{\text{sin}(\pi \alpha)}+(-1)^{k-uN}\cdot\frac{\text{sin}\left[\pi\left(L+q-N\right)\beta\right]}{\text{sin}(\pi \beta)}\right|^2.}
\end{aligned}\end{equation}

Furthermore, the expression of $h(q)$ is
\begin{align}
    \textcolor{myLightPurple}{h(q)=\max\{L-q,2L-N,q-(N-L)\}.}
\end{align}

By combining the above results \textcolor{myLightPurple}{through \eqref{f_g_h_relation}}, we arrive \textcolor{myLightPurple}{at \eqref{block_dp_af_l_da}}. When $e^{j2\pi n\alpha} = 1$ or $e^{j2\pi n\beta} = 1$, the expected squared DP-AF approaches the corresponding limit.

}


\bibliographystyle{IEEEtran}
\bibliography{reference}

\IEEEpubidadjcol

\end{document}